\documentclass{rQUF2e}

\usepackage{float}
\usepackage[utf8]{inputenc}
\usepackage[T1]{fontenc}
\usepackage[english]{babel}

\usepackage{charter}

\usepackage{xcolor}
\usepackage{verbatim}
\usepackage{listings}
\usepackage{multirow}

\usepackage{
	,bbm
	,units 
}
\usepackage{enumerate}

\usepackage{graphicx}
\usepackage{placeins}

\usepackage{
	ulem		
	,cancel		
} 

\numberwithin{equation}{section}
\usepackage[nodayofweek]{datetime}

\usepackage{multirow}
\usepackage{diagbox} 
\usepackage{makecell}

\theoremstyle{plain}
\newtheorem{theorem}{Theorem}[section]

\newtheorem{proposition}[theorem]{Proposition}

\newtheorem{definition}[theorem]{Definition}

\newtheorem{remark}[theorem]{Remark}

\newtheorem{assumption}[theorem]{Assumption}

\newcommand{\1}{\mathbbm{1}}

\newcommand{\bE}{\mathbb{E}}

\newcommand{\cK}{\mathcal{K}}

\newcommand{\cN}{\mathcal{N}}

\newcommand{\cT}{\mathcal{T}}

\newcommand{\dd}{\mathrm{d}}

\begin{document}
\selectlanguage{english}

\title{Capturing Model Risk and Rating Momentum in the Estimation of Probabilities of Default and Credit Rating Migrations}

\author{G. dos Reis${\dag}$${\ddag}$$^\ast$\thanks{Corresponding author.
		Email: G.dosReis@ed.ac.uk}, M. Pfeuffer${\P}$ and G. Smith$\dag$ $\S$\\
	\affil{$\dag$School of Mathematics, Maxwell Institute for Mathematical Sciences, The University of Edinburgh, James Clerk Maxwell Building, Peter Guthrie Tait Road, Edinburgh, EH9 3FD, UK \\
		$\ddag$Centro de Matem\'atica e Aplica\c c$\tilde{\text{o}}$es (CMA), FCT, UNL, Portugal 
		\\
		$\P$University of Erlangen-N\"uremberg, Department of Statistics and Econometrics, Lange Gasse 20, 90403 Nuremberg
		\\
		$\S$Moody's Analytics, 7 Exchange Crescent Conference Square, Edinburgh EH3 8RD		
} }

\maketitle

\begin{abstract}
\noindent
We present two methodologies on the estimation of rating transition probabilities within Markov and non-Markov frameworks. We first estimate a continuous-time Markov chain using discrete (missing) data and derive a simpler expression for the Fisher information matrix, reducing the computational time needed for the Wald confidence interval by a factor of a half. We provide an efficient procedure for transferring such uncertainties from the generator matrix of the Markov chain to the corresponding rating migration probabilities and, crucially, default probabilities.

For our second contribution, we assume access to the full (continuous) data set and propose a tractable and parsimonious self-exciting marked point processes model able to capture the non-Markovian effect of rating momentum. Compared to the Markov model, the non-Markov model yields higher probabilities of default in the investment grades, but also lower default probabilities in some speculative grades. Both findings agree with empirical observations and have clear practical implications.

We use \emph{Moody's proprietary corporate credit rating data set}. Parts of our implementation are available in the R package \emph{ctmcd}.
\end{abstract}
\bigskip
{\bf Keywords:} Confidence Intervals, Markov Chain, Generator Matrix, Point-Process, Rating Momentum
\\ 

\noindent{\bf 2010 AMS subject classifications:}  
Primary: 
60G55 
Secondary: 
62F15 
91G40 
\\



\noindent{\bf JEL subject classifications:} 
G11,
C18 and
G33

%
%

%
%
%
%
%
%
%

\section{Introduction}

Credit risk modelling and financial regulations have received added attention from Mathematics and Economics disciplines since the 2008 financial crash. On January 1, 2018, and for purposes of risk assessment, the new guideline IFRS 9 took effect requiring the calculation of expected losses for the complete maturity of certain obligors' riskier contracts. Thereby, a cornerstone of credit risk modelling lies in the ability to accurately estimate probabilities of default over varying time horizons. This can be either done by considering market data (e.g.~bond or credit default swap prices, as well as implied probabilities of default from equities, see \cite{BieleckiCrepeyHerbertsson2011}) or historical (default or rating) data. In this manuscript, we focus on the latter.

\begin{remark}[Obtaining Default Probabilities: Risk-neutral Vs real-world]
	It is important to note the distinction between these estimation methods. Using market data such as bond prices or credit default swaps to estimate default probabilities actually gives risk neutral default probabilities. Our approach uses observed data and therefore gives real-world (physical) default probabilities. Our results can be used without any adjustment in capital requirement calculations where real-world default probabilities are needed.	
\end{remark}

When estimating probabilities of default, it is typical that credit ratings are considered in the calculation, as they allow for more granularity. Ratings, as categorical solvency measures, might be issued by (external) rating agencies or be produced by the financial institutes themselves as part of the Pillar I internal ratings-based approach underpinning the Basel regulatory framework. Due to idiosyncratic company-level or general business-cycle changes, credit ratings vary over time, and this effect is referred to as a \emph{rating transition} or \emph{rating migration}. This dynamical movement is a stochastic process with a discrete state space in continuous-time. Here Markov chains are a simple, robust and tractable class to model the movement of such rating transitions. The specific models that can be used depend on the type of data available. 

Most literature dealing with the modelling of credit rating transitions focuses on anonymous discrete-time data and often on an annual basis. This data is easier to use and less costly to obtain than the ``full'' (continuous-time company specific rating transitions) data set. In the discretely observed data case, it is not possible to follow individual obligors over the different periods which forces one to treat all companies in the same rating as equivalent. Hence, one is naturally led to a Markov chain construct (continuous or discrete). Assuming a continuous-time Markov chain (CTMC) model that has been observed only at specific discrete points, obtaining the maximum likelihood estimator (MLE) and understanding the error of this estimate is a classical problem. Many works have investigated the estimation of Generator matrices or the (intermediate) Transition Probability Matrices (TPM), see \cite{KalbfleischLawless1985}, \cite{BladtSorensen2005}, \cite{BladtSorensen2009}, \cite{ReisSmith2017} to mention a few and \cite{Pfeuffer2017-R-Package} for an overview and algorithm implementation in the statistical language R.

Despite these works, the problem of how to conduct statistical inference in this context or, in particular, how to derive error estimates for discretely observed Markov processes, is still an issue. Since our inference is likelihood based, Wald confidence intervals (or Wald intervals) are the natural choice for error estimation. In \cite{KalbfleischLawless1985}, the authors use numerical techniques to estimate the derivatives and use a so-called \emph{quasi-Newton} method to obtain the MLE\footnote{A quasi-Newton method (or scoring procedure) only requires one to estimate the first order derivative. More common approaches such as Newton-Raphson for finding the MLE would also require evaluation of the second derivative.}. More recently, \cite{BladtSorensen2005} and \cite{BladtSorensen2009} consider the Expectation Maximisation (EM) algorithm and a Markov chain Monte Carlo (MCMC) algorithm to obtain the MLE. In the case of the EM, the authors provide a numerical scheme based on a formula from \cite{Oakes1999} to obtain the error in the estimate. Following their approach \cite{ReisSmith2017} give exact expressions for errors arising from the EM algorithm. Building on these, we transfer the errors in the estimation of the CTMC's generator matrix to the estimation errors of the rating transitions probabilities themselves. As far as we are aware, such estimations have not been considered, although, they are of significant practical importance. When complete continuous-time rating transition data is available, then the computation of point estimates and Wald intervals for the parameters of a CTMC is straightforward, see e.g.~\cite{LandoSkodeberg2002}.

Concerning the second contribution of our manuscript, \cite{LandoSkodeberg2002} show that rating transitions exhibit non-Markovian behaviours. In particular, an obligor that has been recently downgraded into a certain rating is more likely to be downgraded further than other obligors currently in that rating. Such an effect is referred to as \emph{(downward)-rating momentum}. A similar effect may also appear in upgrades; however, it is not as apparent. Documented (non-Markovian) effects in rating transitions include \emph{rating drift} (or \emph{momentum}) in \cite{AltmanKao1992} and \cite{LandoSkodeberg2002}, \emph{rating stickiness} in \cite{McNeilFreyEmbrechts2005} and specific rating agencies' policies (see \cite{CareyHrycay2001} and \cite{Loffler2005}). \cite{NickellPerraudinVarotto2000} highlight non-Markovian patterns in transition probabilities for ratings and discuss their dependence regarding underlying variables like industry, domicile and business cycle. However, of these effects rating momentum is the most important to capture and what we look to model here.

The rating momentum effect has a non-negligible bearing on the risk attributed to a portfolio as it makes defaults of investment grade bonds likelier than defaults that are estimated within the standard Markov framework.  
\cite[p.8]{couderc2008credit} report on the temporal span of the rating drift (for a certain  Standard \& Poor’s database) and its mean reversion. When looking to model over longer horizons, the non-Markov effects such as momentum become more pronounced, i.e.~have a larger impact on transition probabilities. At a practical level, the IFRS9 regulation requires knowledge of risks on rating migrations over longer horizons where these effects can significantly change the results. When one can access the full data set (continuous-time observations), it is possible to construct tractable models that capture non-Markov effects and this is one of our contributions. The model we propose is able to capture the momentum behaviour, and we found that the purely Markov model underestimates default risk in investment grades but overestimates the risk in some speculative grades. We discuss this in more detail in Section \ref{Sec:Rating Momentum Model} below. 

For clarity, we summarize the contributions of our manuscript in the next two points.
\begin{enumerate}
	\item In the CTMC setting with discretely observed data, we provide a new simpler closed-form expression for the Hessian of the likelihood function, enabling faster computation of confidence intervals via the Fisher information matrix (Wald intervals). We further provide expressions allowing one to transfer confidence intervals at the level of the generator matrix to the level of rating transitions and probabilities of default, where they can be easily interpreted. Recall that in the case of discrete anonymous data one is forced to adopt the Markov assumption.
	
	\item In the setting of continuously observed data, we propose a tractable and parsimonious model that captures the non-Markovian phenomenon of \emph{rating momentum}. We provide a calibration procedure and several comparative tests based on \emph{Moody's Corporate Credit Ratings} data set (see Section \ref{Sec:Data and Testing}). Most notable is the difference between empirical, Markov (CTMC) and non-Markov (our model) estimates of probabilities of default. We observe that in several cases the Markov model under- or overestimates the probabilities of default empirically observed while the non-Markov model provides better agreement.
\end{enumerate}

\begin{remark}[Software and R-code]
	\label{rem:pfeuffer-ctmcd}
	The algorithms relating to Markov Chains (in Section \ref{Sec:EM Error}) are part of the CRAN R-package \emph{ctmcd: Estimating the Parameters of a Continuous-Time Markov Chain from Discrete-Time Data} (see \cite{Pfeuffer2017-R-Package}) --- {https://CRAN.R-project.org/package=ctmcd}
\end{remark}

\subsubsection*{Potential Non-Markov Models with Application to Rating Momentum.}
Different models have been introduced in the past to incorporate non-Markov phenomena. We briefly overview some of these works here. Interested readers are encouraged to consult the reference herein.

\emph{Extended State Space and Mixture Models.}
\cite{ChristensenEtAl2004}, attempt to take non-Markovian effects into account while preserving some Markovian structure. The idea is to extend the state space to include $+$ and $-$ states, referred to as excited states. For example, when a company downgrades from rating $A$ to rating $B$, it is instead given the rating $B^{-}$, which has a higher probability of further downgrades than $B$. Similarly, if the company transitions from $B$ to $A$, it is instead rated $A^{+}$ which has a smaller probability of downgrade than $A$. 
This construction allows us to maintain the Markov property; however, we must calibrate many more parameters and, in real-world data,  we do not observe a company belonging to the excited or non-excited state. Moreover, when successive transitions occur, it is unknown whether the company was in the excited or non-excited state. Hence calibrating an intensity between excited and non-excited states seems impossible. One could navigate around this by assuming excited states do not jump to non excited states, but this is against empirical evidence of momentum reducing over time, see \cite[p. 35]{couderc2008credit} for example.

\cite{DAmico2016} apply a semi-Markov model to capture the observed effect that companies move through states not following an exponential distribution. However, they still rely on the Markov transition structure and hence they need to expand the state space in order to include momentum. Related to this approach is \cite{FrydmanSchuermann2008}, where the authors cleverly use two different time homogeneous CTMC generator matrices, however, it does not capture momentum since the jump itself is Markov.

\emph{Hidden Markov Model (HMM).}
A different idea is to use a hidden Markov model (HMM) (see \cite{CappeEtAl2005} for a complete account). The HMM approach to credit risk can be traced back to the work of R.~Elliot, see overview in \cite{Korolkiewicz2012} and its references. Roughly, the approach considers two processes $(X,Y)$, the observed (published) credit rating \emph{Y} and the ``true" credit rating \emph{X} which is unobserved (or hidden). The paradigm is that the observed credit ratings are assumed to be ``noisy'' observations of $Y$ and not the true representation of the credit risk. The goal is then to use $Y$ to make inference on $X$. In such a setup if one considers the noisy observation and the true rating as correlated, then rating momentum can be added into the model. Although this approach has some benefits, the work appears to be constrained to the discrete time case and, from the literature point of view, the approach remains unexplored.

\emph{Hazard Rates, Point Processes and self-exciting Marked Point Processes.} Let us start by discussing Hazard rates. The main work in this area for credit ratings is given in \cite{KoopmanEtAl2008}. An extensive work bringing hazard rate methodologies to the estimation of probabilities of default can be found in \cite{couderc2008credit} (and references therein). The paradigm is that each company has a corresponding hazard rate (a parameter) and in this hazard rate one can encode various factors such as momentum for example. The issue with \cite{KoopmanEtAl2008}'s methodology is that they must calibrate parameters for each of the various transitions with the extra variables to obtain the probabilities of these transitions. This however, increases the model's complexity greatly. Our goal is to present a model as parsimonious as possible that captures rating momentum.

The approach we pursue relies on \emph{point processes} that are dependent on their own history, so-called \emph{self-exciting} processes (see \cite{DaleyVereJones2003}, \cite{DaleyVereJones2007}). Point processes are generalizations of Markov processes and a natural choice for our model. One of the most satisfying aspects of using point processes is that one can capture rating momentum by adding only a small number of parameters ($2$ to $4$ in our case). The most common example of a self-exciting process is the Hawkes process. These processes appear in other areas of mathematical finance, such as models for limit order books and are also used in high-frequency trading, see \cite{BacryMastromatteoMuzy2015}. However, they have not been fully utilized in credit transitions. 
A Hawkes process can be thought of as a counting process (similar to a Poisson process) which in one dimension has an intensity $\lambda_{t}$ of the form (see \cite{DassiosZhao2013}),
\begin{align*}
\lambda_{t} =
\mu
+
\int_{0}^{t} \phi(t-s) \dd N_{s} \, ,
\end{align*}
where $N$ is a counting measure and denotes that an event has occurred (this will be a rating change in our case), $\mu$ is the baseline intensity and $\phi$ is the impact on the intensity and allows the intensity to depend on past events. 
By setting $\phi =0$ the Hawkes process reduces to a Poisson process. A common choice for $\phi$ is the so-called exponential decay, namely $\phi(t-s)=\alpha \beta \exp(-\beta(t-s) )$ with $\alpha,\beta>0$. Functions of this form are useful since the event's influence on the intensity weakens as time progresses, hence we can account for momentum reducing over time (agreeing with the findings of \cite{couderc2008credit}).

Using a Hawkes process allows us to embed past dependence in the jump times, however, in this simplistic form it is not fit for our purposes since we require different changes to intensity dependent on whether it is an upgrade or a downgrade. Further, we require the baseline intensity $\mu$, to depend on the current state. Such extended Hawkes processes are referred to as \emph{marked point processes}, since to each event observed one assigns a \emph{mark} to indicate the type of event, see \cite[Chapter 6.4]{DaleyVereJones2003}. We discuss this further in Section \ref{Sec:Rating Momentum Model}.
\medskip

This work is organized as follows. In Section \ref{Sec:Data and Testing} we overview the data paradigms and describe the data we work with. In Section \ref{Sec:EM Error} we establish our closed-form expression for the Wald confidence intervals for the underlying Transition Probability Matrix (TPM) in the Markov setting. Finally, in Section \ref{Sec:Rating Momentum Model} we analyse Moody's corporate credit rating data set, we test for non-Markovianity and calibrate the proposed non-Markov model. We also give due attention and discuss the effect of adding momentum in the estimation of default probabilities. In order to help keep this work self contained we supplement the core ideas with further discussion in the Appendix.

%
%
%
%
%

\section{Data description}
\label{Sec:Data and Testing}

To illustrate the statistical methods we develop in this manuscript, we use the proprietary \emph{Moody's corporate credit rating data set}, which comprises continuous-time observations for 17097 entities (companies) in the time interval Jan 1, 1987 to Dec 31, 2017. Through the remainder of the article we refer to this set as the ``\emph{Moody's data set}''. Some of the discrete data is available publicly, but the full data set is proprietary and must be purchased. Other works such as \cite{ChristensenEtAl2004} also use the full Moody's data set.

The rating categories in Moody's data set are depicted in decreasing order of rating quality as ``Aaa'', ``Aa1'', ``Aa2'', ``Aa3'', ``A1'', ``A2'', ``A3'', ``Baa1'', ``Baa2'', ``Baa3'', ``Ba1'', ``Ba2'', ``Ba3'', ``B1'', ``B2'', ``B3'', ``Caa1'', ``Caa2'', ``Caa3'', ``Ca'', ``C''. We define ``C'' as the default category. The refinements ``1'', ``2'' and ``3'' shall be referred to as \emph{modifiers} in the following. The ratings ``Aaa'' to ``Baa3'' are the so-called ``Investment Grade'' block while the ratings ``Ba1'' to ``Ca'' form the ``Speculative Grade'' block. 

We employ a standard data aggregation arrangement where we aggregate all modifiers within their rating class. For instance, we group ``Aa1'', ``Aa2'', ``Aa3'' as ``Aa'' and so on to obtain the following categories in decreasing credit quality: ``Aaa'', ``Aa'', ``A'', ``Baa'', ``Ba'', ``B'', ``Caa'', ``Ca'' and ``C'' (Default Category). 
We shall use the standard aggregation unless otherwise stated. 

For clarification, unlike the Standard and Poor rating classes where ``C'' is taken as the rating above default and ``D'' is used as default, in the Moody's rating system, ``C'' denotes default. We use the latter notation throughout our manuscript.

As described in the introduction there are two data paradigms, a \emph{discrete} (missing) and a \emph{continuous} (full) one. In Section \ref{Sec:EM Error} of the paper we construct annually discretized rating transition matrices from this data, and one is led to use a (CTMC) Markov model. In Section \ref{Sec:Rating Momentum Model}, we use the full data set and its richness allows us to expand the scope to non-Markov models.

%
%
%
%

\section{Calculating Wald Confidence Intervals for Discretely Observed Markov Processes}
\label{Sec:EM Error}

The working paradigm for this section is the \emph{discrete time data} one and we work towards estimating the generator matrix $\mathbf{Q}$ of the underlying CTMC model. For this setting, it was shown in \cite{ReisSmith2017} that the Expectation-Maximization (EM) algorithm is the strongest algorithm for the estimation of $\mathbf{Q}$ (a description of the EM algorithm for this context is provided in Appendix \ref{Sec:Fundamentals of Markov}). The EM is built to use likelihood-based inference which has the advantage that one can obtain errors for the estimate by taking derivatives, the so-called Wald confidence intervals. The goals of this section are to find expressions for these derivatives and then use them to obtain the corresponding intervals for the transition probabilities.

Our CTMC set up is similar to that of \cite{ReisSmith2017}.
We understand companies' ratings as defined on a finite state space $\{1, \dots, h\}$, where each state corresponds to a rating. We denote $Aaa$ as rating $1$ and $C$ (default) as rating $h$. Let $\mathbf{P}$ be an $h$-\emph{by}-$h$ stochastic matrix, which will be the corresponding TPM (at, say, time $t=1$) and $\mathbf{Q}$ is an $h$-\emph{by}-$h$ generator matrix; we denote $p_{ij}:=\left(\mathbf{P}\right)_{ij}$, $q_{ij}:=\left(\mathbf{Q}\right)_{ij}$ and the intensity of state $i$ by $q_{i}=\sum_{j \neq i}q_{ij}$ where $i,j \in \{1,\dots,h\}$. A standard assumption used in credit risk modelling is that the default state is an absorbing state, hence $p_{hh}=1$. In the data, companies are observed to withdraw (e.g.~via mergers or early payment) and we treat such a withdrawn rating as a censored result.

Regarding the CTMC's generator, we work with stable generator matrices, i.e.~matrices $\mathbf{Q}$ that satisfy the following definition.
\begin{definition}[Stable-Conservative infinitesimal Generator matrix of a CTMC]
	\label{def:GeneratorMatrixSpace}
	We say a matrix $\mathbf{Q}$ is a generator matrix if the following properties are satisfied for all $i,j \in \{1, \dots, h\}$:
	\\ 
${}$\qquad  i)  $0 \le q_{ij} < \infty$ for $i \neq j$;\qquad  ii) $q_{ii} \le 0$;\qquad  and iii) $\sum_{j=1}^{h} q_{ij}=0$.
\end{definition}
Our quantity of interest is the time varying transition probability matrix, $\mathbf{P}(t)$, which is related to the generator matrix $\mathbf{Q}$ via,
\begin{align}
\label{Eq:P and Q relation}
\mathbf{P}(t) = e^{\mathbf{Q}t},\qquad t\geq 0 .
\end{align}

We assume throughout that $\mathbf{Q}$ is a valid generator matrix (in the sense of Definition \ref{def:GeneratorMatrixSpace}), hence $\mathbf{P}$ is well defined. 
Considering the case where the CTMC is observed at times $t_{0}<t_{1} < \dots < t_{M}$ and denote $\Delta t_{u} := t_{u} - t_{u-1}$ for $u \in \{1, \dots, M \}$ and the transition matrix over that interval by $\mathbf{N} (u)$.

The likelihood of the discretely observed Markov process is given by,
\begin{align}
\label{eq:LikelihoodCTMC}
L(\mathbf{Q}| \mathbf{N})= \prod_{u=1}^{M} \prod_{s=1}^h \prod_{r=1}^h\exp(\mathbf{Q} \Delta t_{u})_{sr}^{\mathbf{N}_{sr}(u)}.
\end{align}
Although this is not the full likelihood of a CTMC, it is the likelihood based on the observable data, so in effect, the EM algorithm looks to find $\mathbf{Q}$ to maximise \eqref{eq:LikelihoodCTMC}. Therefore the Wald confidence intervals of $\mathbf{Q}$ are based on this likelihood.
One can construct confidence intervals for other algorithms such as the quasi-optimization of the generator (see \cite{KreininSidelnikova2001}) by bootstrapping, but these are computationally more expensive to calculate.

\subsection{Direct Differentiation for Gradient and Hessian of the Likelihood}
\label{Sec:Direct Differentiation}
The standard procedure to derive a confidence interval is to use the variance of the estimator (in our case, the negative inverse of the Hessian $H$ of the likelihood $L$ in \eqref{eq:LikelihoodCTMC}). Since the EM algorithm deals with a missing data likelihood, these derivatives are complex to calculate, however, \cite[Section 2]{Oakes1999} derived a simpler formula for the Hessian. This formula was used by \cite{BladtSorensen2009} and \cite{ReisSmith2017} to obtain error estimates in this setting. 
A formula for obtaining the Hessian is useful, however, while the second derivative can inform us about errors at the level of the generator matrix, it does not shed light on how these errors propagate to the transition probabilities (see \eqref{Eq:P and Q relation}). For that we need to be able to take further derivatives.

Relying on first principles, it turns out that for this problem one can extract derivatives without the said formula in \cite{Oakes1999} and derive a new closed-form solution involving matrix exponentials for the gradient and the Hessian by direct differentiation.

Similar to the situation in \cite{ReisSmith2017} the parameter space of $\mathbf{Q}$ is closed at zero and we can only differentiate in the interior of the space, hence we introduce the notion of \emph{allowed pairs}. This concept allows one to incorporate absorbing states in the analysis.
\begin{definition}[Allowed pairs]
	\label{def:AllowedPairs}
	Let $i, j \in \{1, \dots, h\}$, then we say that the pair $(i, j)$ is \emph{allowed} if $i \neq j$ (not in the  diagonal) and $q_{ij}$ is not converging to zero under the EM algorithm.
\end{definition}
Essentially $i$, $j$ is allowed if $q_{i j}>0$ and thus in the interior of the parameter space of $\mathbf{Q}$. For ease of presentation we denote by $\mathbf{V}_{\mathbf{Q}}$ the matrix of allowed pairs of $\mathbf{Q}$, namely: for $N_{a}$ the number of allowed pairs in the estimation of $\mathbf{Q}$ we define the matrix $\mathbf{V}_{\mathbf{Q}}$ as the  $N_{a}$-\emph{by}-$2$-dimensional matrix which records the allowed pairs of $\mathbf{Q}$.

Let $\mathbf{A}$ be an $h$-by-$h$ matrix, $\alpha$, $\beta$, $s$, $r$ $\in \{1,\dots, h\}$ and $\mathbf{e}_{r}$ be an $h$-dimensional column vector with $1$ at entry $r$ and zero elsewhere. Let us further denote $a_{\alpha \beta} := (\mathbf{A})_{\alpha \beta}$ as the entries of matrix $\mathbf{A}$ and assume $a_{\alpha \beta} >0$, then using standard properties of derivatives and integrals of matrix exponentials (see \cite{Wilcox1967} and \cite{VanLoan1978}) it follows that,
\begin{equation}
\label{Eq:Matrix Exponential Derivative Relation}
\frac{\partial \exp(\mathbf{A}t)_{sr}}{\partial a_{\alpha\beta}}=\mathbf{e}_{s}^{\intercal}\int_0^{t}\exp(\mathbf{A}v)\frac{\partial \mathbf{A}}{\partial a_{\alpha\beta}}\exp\big(\mathbf{A}(t-v)\big) \dd v \mathbf{e}_{r}
= \mathbf{e}_{s}^{\intercal}\exp\left(\left[\begin{smallmatrix}
\mathbf{A} & \frac{\partial \mathbf{A}}{\partial a_{ij}} \\
{0} & \mathbf{A}
\end{smallmatrix}\right] t \right)_{1:h, h+1:2h} \mathbf{e}_{r}.
\end{equation} 
Using \eqref{Eq:Matrix Exponential Derivative Relation}, we can directly calculate the first and second derivative of the likelihood function for a discretely observed Markov process. 
Let $(\alpha, \beta)$ and $(\mu, \nu)$ be allowed pairs for the generator $\mathbf{Q}$, then the expressions for the gradient and Hessian of the logarithm of \eqref{eq:LikelihoodCTMC} are as follows: for the \textbf{Gradient} we have
\begin{align*}
\frac{\partial \log L(\mathbf{Q}| \mathbf{N})}{\partial q_{\alpha\beta}}
=&
\sum_{u=1}^M\sum_{s=1}^h\sum_{r=1}^h
  \mathbf{N}_{sr}(u)\frac{\exp(\mathbf{C}_\eta^{(\alpha\beta)} \Delta t_{u})_{s,h+r}}{\exp(\mathbf{Q} \Delta t_{u})_{s,r}} \quad \text{ with } \quad \mathbf{C}_\eta^{(\alpha\beta)}=\begin{bmatrix}
                  {\mathbf{Q}} ~ & \mathbf{e}_\alpha \mathbf{e}_\beta^{\intercal}-\mathbf{e}_\alpha \mathbf{e}_\alpha^{\intercal} \\
                  {0} & \mathbf{Q}\end{bmatrix},
\end{align*}
while for the \textbf{Hessian} we have
\begin{align*}
& H(\mathbf{Q})_{\alpha\beta,\mu\nu}
=
\frac{\partial^2 \log L(\mathbf{Q}|\mathbf{N})}{\partial q_{\alpha\beta}\partial q_{\mu\nu}}
\\
&=
\sum_{u=1}^M\sum_{s=1}^h\sum_{r=1}^h \frac{\mathbf{N}_{sr}(u)}{\exp(\mathbf{Q} \Delta t_{u})_{sr}}
\left[\frac{\exp(\mathbf{C}_\eta^{(\alpha\beta)} \Delta t_{u})_{s,h+r}\exp(\mathbf{C}_\eta^{(\mu\nu)} \Delta t_{u})_{s,h+r}}{\exp(\mathbf{Q} \Delta t_{u})_{sr}}
-
\exp(\mathbf{C}_\xi^{(\alpha\beta,\mu\nu)} \Delta t_{u})_{s,3h+r}
\right]
,
\\
&\qquad \qquad \text{whereas } \mathbf{C}_\xi^{\alpha\beta,\mu\nu}=\begin{bmatrix}
                  \mathbf{C}_\eta^{(\alpha\beta)} & \frac{\partial \mathbf{C}_\eta^{(\alpha\beta)}}{\partial q_{\mu\nu}} \\
                  {0} & \mathbf{C}_\eta^{(\alpha\beta)}\end{bmatrix}.
\end{align*}
These estimates are direct applications of the theory above and hence we omit the steps.
Both the formula of \cite[p. 7]{ReisSmith2017} and this new one are exact expressions for the Hessian and thus for the Fisher information matrix. However, the new formula is of distinctly reduced complexity, which consequently leads to clearly shorter computing times. 
Since the Hessian is only defined for allowed pairs the matrix is dimension-wise smaller than $(h-1)^{2}$-\emph{by}-$(h-1)^{2}$. 

We compute the Wald confidence intervals as follows,
\begin{itemize}
	\item recall that $\mathbf{V}_{\mathbf{Q}}$ is the $N_{a}$-\emph{by}-$2$ dimensional matrix recording the allowed pairs of $\mathbf{Q}$ (with $N_{a}$ the number of allowed pairs in the estimated $\mathbf{Q}$).

The $ij${th} component of the Hessian is the differential
	\begin{equation*}
	\frac{\partial^{2}}{\partial q_{\mathbf{V}_{\mathbf{Q}}(i,1) \mathbf{V}_{\mathbf{Q}}(i,2)} \partial q_{\mathbf{V}_{\mathbf{Q}}(j,1) \mathbf{V}_{\mathbf{Q}}(j,2)}} \, .
	\end{equation*}
	
	\item The Fisher information matrix is given by $-\mathbf{H}(\cdot)$. The estimated variance of the allowed parameter $q_{ab}$ is the $i^{th}$ diagonal element of  $-\mathbf{H}(\cdot)^{-1}$, where $\mathbf{V}_{\mathbf{Q}}(i,1)=a$ and $\mathbf{V}_{\mathbf{Q}}(i,2)=b$.
	
	\item The Wald $95\%$ confidence interval of the MLE $\hat{q}_{ab}$ is $\hat{q}_{ab} \pm 1.96\sqrt{Var(\hat{q}_{ab})}$.
\end{itemize} 

A $95$\% confidence interval for the generator matrix estimate based on Moody's discretely observed corporate rating data is illustrated in Table \ref{Table:95 Wald for Full Data}. To obtain the Wald confidence interval, the computation time was $\approx 1$s with the new expression compared to $\approx 2$s for the formula of \cite{ReisSmith2017}.

\begin{table}[!h]
	\hspace{-0.5cm}
	\scriptsize
	\begin{tabular}{ c | c | c | c | c | c | c | c | c }
		 Aaa & Aa & A & Baa & Ba & B & Caa & Ca & C \\ \hline
		~ & [0.074,0.091] & 0 & 0 & [0,0.001] & 0 & 0 & 0 & 0 \\ \hline
		[0.009,0.012] & ~ & [0.088,0.098] & [0.001,0.004] & [0,0.001] & 0 & 0 & 0 & 0 \\ \hline
		[0,0.001] & [0.023,0.027] & ~ & [0.061,0.067] & [0.003,0.005] & [0.001,0.002] & [0, 0.001] & 0 &0 \\ \hline
		[0,0.001] & [0.001,0.002] & [0.039,0.044] & ~ & [0.042,0.047] & [0.005,0.007] & [0.001,0.003] & [0,0.001] & 0 \\ \hline
		0 & [0,0.001] & [0.002,0.004] & [0.064,0.072] & ~ & [0.092,0.102] & [0.007,0.011] & [0.001,0.002] & 0 \\ \hline
		0 & [0,0.001] & [0,0.001] & [0.001,0.003] & [0.049,0.055] & ~ & [0.091,0.099] & [0.008,0.011] & 0 \\ \hline
		0 & 0 & 0 & [0,0.001] & [0.001,0.005] & [0.107,0.122] & ~ & [0.052,0.064] & [0.028,0.036] \\ \hline
		0 & 0 & 0 & [-0.001,0.006] & [0.003,0.018] & [0.047,0.083] & [0.127,0.181] & ~ & [0.123,0.170] \\ \hline
		0 & 0 & 0 & 0 & 0 & 0 & 0 & 0 & 0 
		\\ \hline
	\end{tabular}
	\caption{Confidence Interval (at 95 \% confidence) for the entries of the Generator Matrix for Moody's Corporate Rating Discrete-Time Transition Matrix.}
	\label{Table:95 Wald for Full Data}
\end{table}

\subsection{The Delta method - Confidence Intervals for probabilities}

The object we are estimating is the generator matrix $\mathbf{Q}$, thus the confidence intervals are based on the entries of this matrix. Although obtaining these confidence intervals are useful, from a practitioners standpoint it is more useful to know how this uncertainty propagates to the underlying TPM and the estimated probabilities of default. This is a classical problem in statistics where one wishes to consider how the confidence interval changes under a transformation (in this case \eqref{Eq:P and Q relation}), the standard method to do this is known as the \emph{Delta method}, see \cite{LehmannCasella1998} for further information.

We construct confidence intervals for each individual element in $\mathbf{P}$ using the set of \emph{allowed pairs} (Definition \ref{def:AllowedPairs}). We consider the confidence interval for the transition probability $p_{ij}$ at time $t$ as,
\begin{align*}
p_{ij}(\mathbf{V}_{\mathbf{Q}} ; t) :=
\left( e^{\mathbf{Q}t} \right)_{ij} \, .
\end{align*}
That is for a fixed $t$, $p_{ij}(\mathbf{V}_{\mathbf{Q}} ; t)$ is a multivariate function of the allowed pairs, $\mathbf{V}_{\mathbf{Q}}$, in $\mathbf{Q}$.
This leads to the following result.

\begin{theorem}
	\label{Thm:CTMC Delta Method}
	Assume asymptotic normality holds for all allowed pairs, let $\mathbf{V}_{\hat{\mathbf{Q}}}$ denote the allowed pairs of $\hat{\mathbf{Q}}$ (our MLE estimate) and fix $t$. Then, for each $i,j$ in the state space with $i \neq h$, the variance in $p_{ij}$ is given by,
	\begin{align}
	\label{Eq:Variance estimate}
	\textrm{Var}\Big(\,p_{ij}(\mathbf{V}_{\hat{\mathbf{Q}}} ; t)\,\Big)
	\approx
	\frac{\partial p_{ij}(\mathbf{V}_{\hat{\mathbf{Q}}} ; t)}{\partial \mathbf{V}_{\hat{\mathbf{Q}}}}
	\left( -\mathbf{H}(\hat{\mathbf{Q}})^{-1} \right)
	\left(\frac{\partial p_{ij}(\mathbf{V}_{\hat{\mathbf{Q}}} ; t)}{\partial \mathbf{V}_{\hat{\mathbf{Q}}}} \right)^{\intercal} \, ,
	\end{align}
	provided $\partial p_{ij}(\mathbf{V}_{\hat{\mathbf{Q}}} ; t)/ \partial \mathbf{V}_{\hat{\mathbf{Q}}} \neq 0$,
	where $\frac{\partial}{\partial \mathbf{V}_{\hat{\mathbf{Q}}}} $ denotes the vector constructed by differentiating w.r.t.~each element in $\mathbf{V}_{\hat{\mathbf{Q}}}$ then evaluated at $\hat{\mathbf{Q}}$, and $\mathbf{H}(\hat{\mathbf{Q}})^{-1}$ is the inverse Hessian matrix at the MLE. Moreover, for each $(\alpha, \beta) \in \mathbf{V}_{\hat{\mathbf{Q}}}$,
	\begin{align*}
	\frac{\partial p_{ij}(\mathbf{V}_{\hat{\mathbf{Q}}} ; t)}{\partial q_{\alpha \beta}}
	=
	\left( \exp ( \mathbf{C}_{\eta}^{(\alpha \beta)} t) \right)_{i, h+j}
\quad\textrm{where}\quad
	\mathbf{C}^{(\alpha \beta)}_{\eta}=
	\left[ {\begin{array}{cc}
		\hat{\mathbf{Q}} & \mathbf{e}_{\alpha} \mathbf{e}_{\beta}^{\intercal} - \mathbf{e}_{\alpha} \mathbf{e}_{\alpha}^{\intercal}\\
		0 & \hat{\mathbf{Q}}
		\end{array} } \right] \, .
	\end{align*}
\end{theorem}

The proof of this result is given in Appendix \ref{Sec:Proof of Delta Method}.
The assumption that $\partial p_{ij}(\mathbf{V}_{\hat{\mathbf{Q}}} ; t)/ \partial \mathbf{V}_{\hat{\mathbf{Q}}} \neq 0$ is extremely mild and can be easily checked once the MLE estimate is found.  

At this point, we take advantage of the fact that we have already derived a closed-form expression for the Hessian. Hence we can easily compute \eqref{Eq:Variance estimate}, moreover, it is now straightforward to compute the confidence interval for the transition probabilities. This is an extremely useful result since it allows one to quantify the uncertainty at the level of the estimation of transition probabilities (instead of the generator matrix), and critically, uncertainties in the probability of default. Figures \ref{g2-1} and \ref{g2-2} show such intervals for probability of default estimates from Moody's corporate rating data 2016 and a time horizon of up to 10 years. One can see that this procedure easily allows one to quantify the error of probability of default predictions for arbitrary time horizons. This is especially interesting as this parameter is an important ingredient to the calculation of expected losses over lifetime in the IFRS 9 regulatory framework.

\begin{figure}[!ht]
\centering
\includegraphics[width=.47\textwidth]{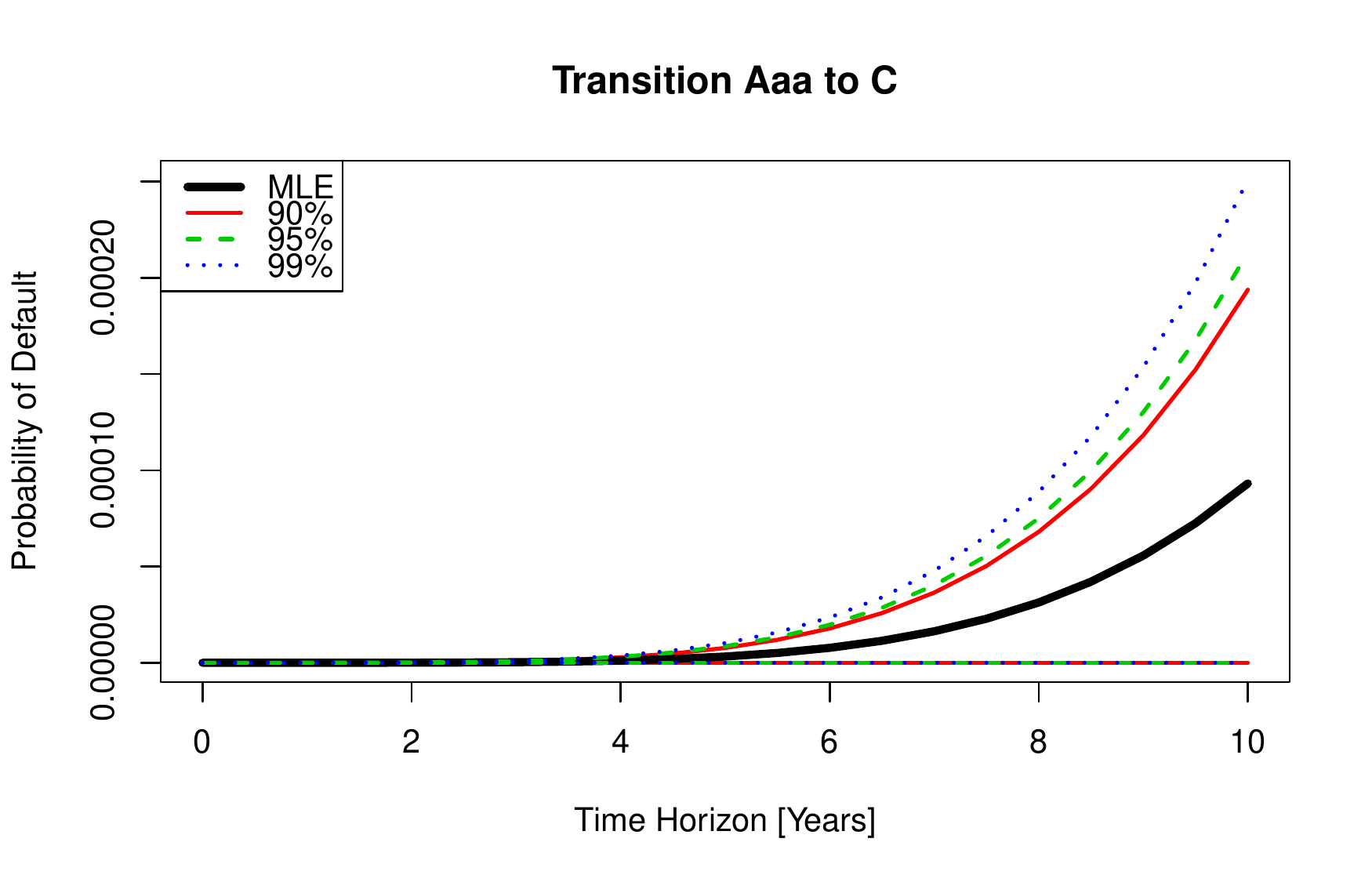}\includegraphics[width=.47\textwidth]{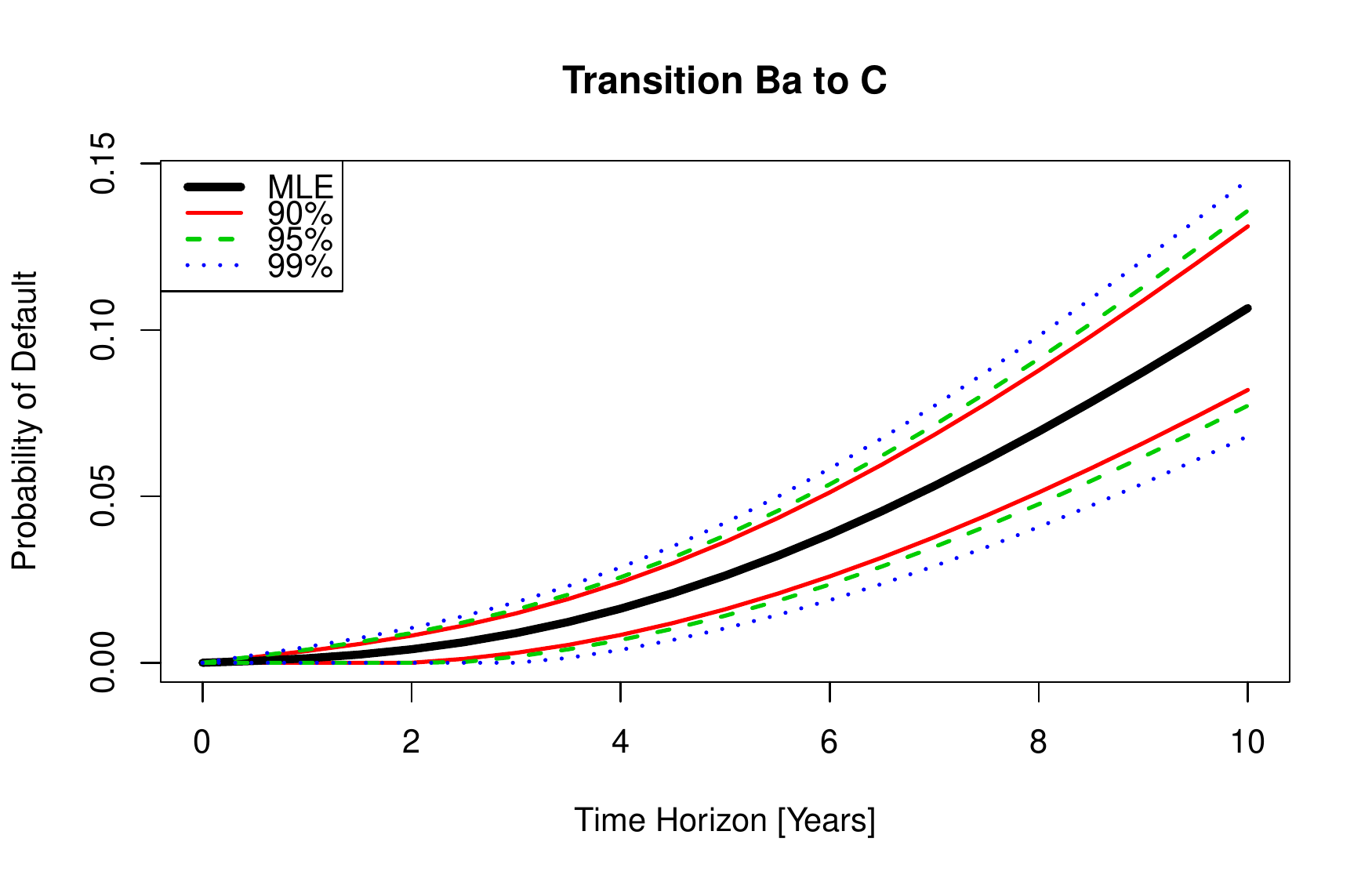}\\
\includegraphics[width=.47\textwidth]{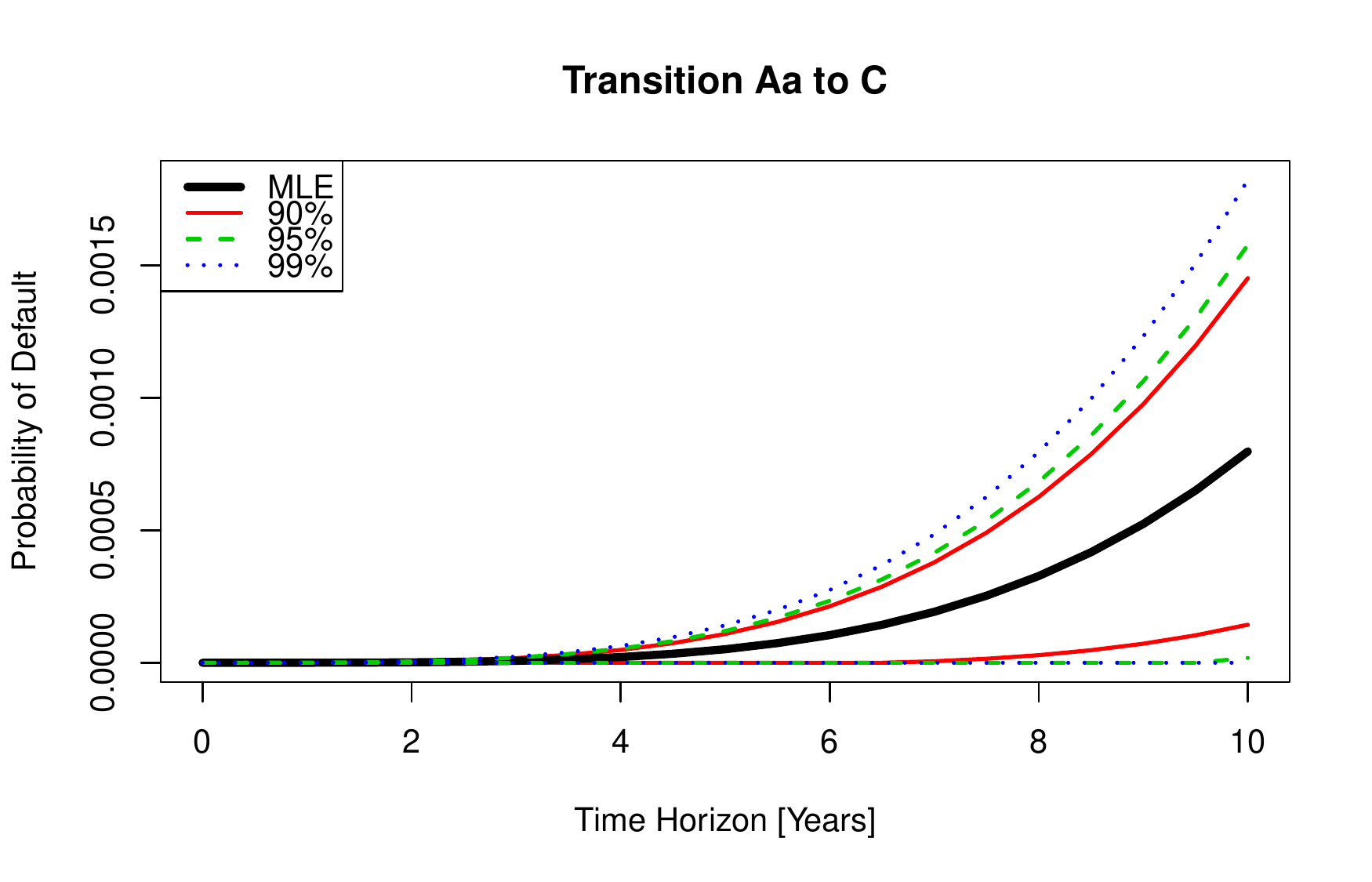} \includegraphics[width=.47\textwidth]{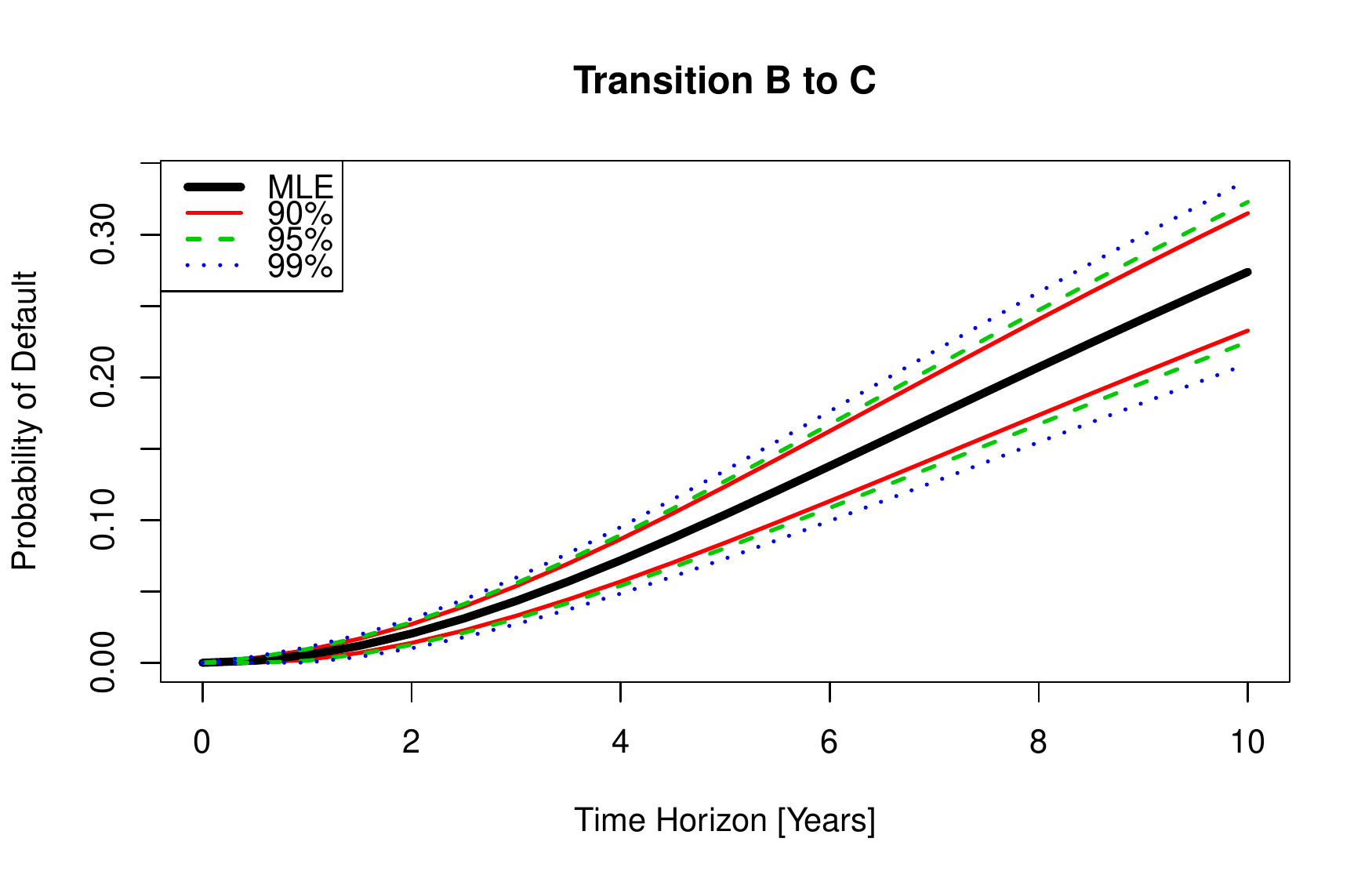}
\caption{Confidence Intervals as maps of time for Discrete-Time Transitions into the Default Category $C$ over 10 years-  Moody's Corporate Rating Discrete-Time Transitions 2016
}
\label{g2-1}
\end{figure}

\begin{figure}[!ht]
\centering
\includegraphics[width=.47\textwidth]{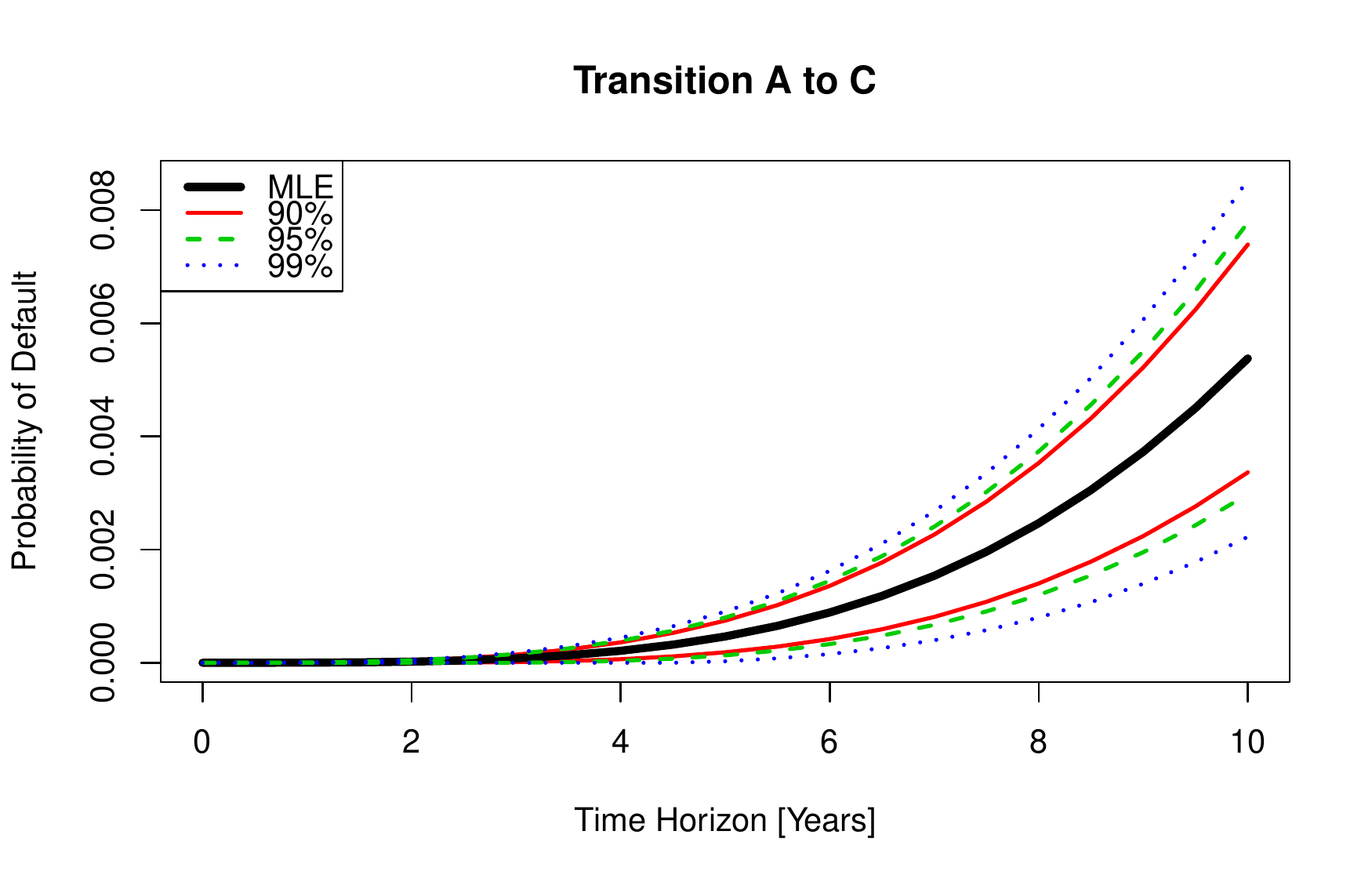}\includegraphics[width=.47\textwidth]{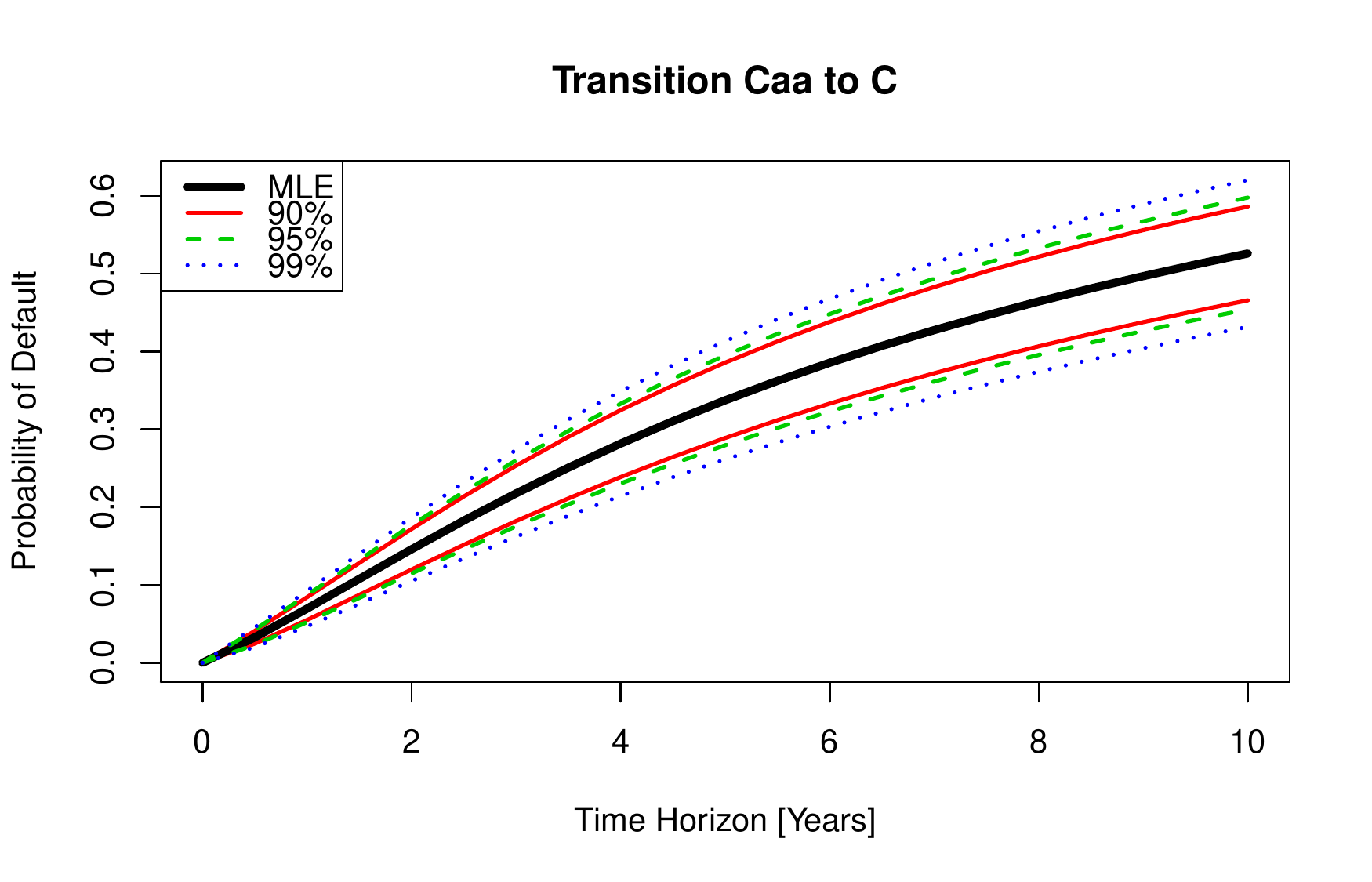}\\
\includegraphics[width=.47\textwidth]{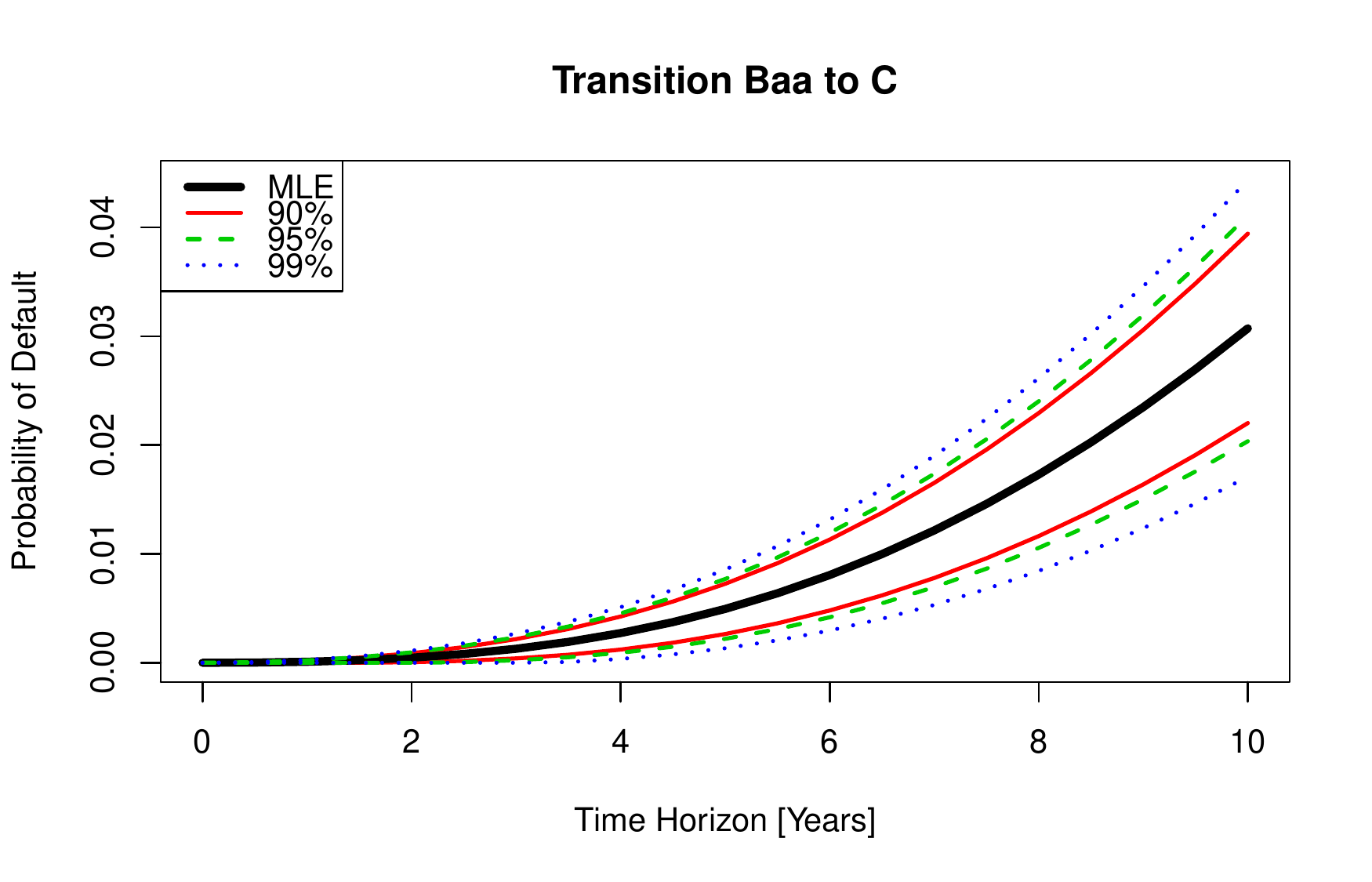}\includegraphics[width=.47\textwidth]{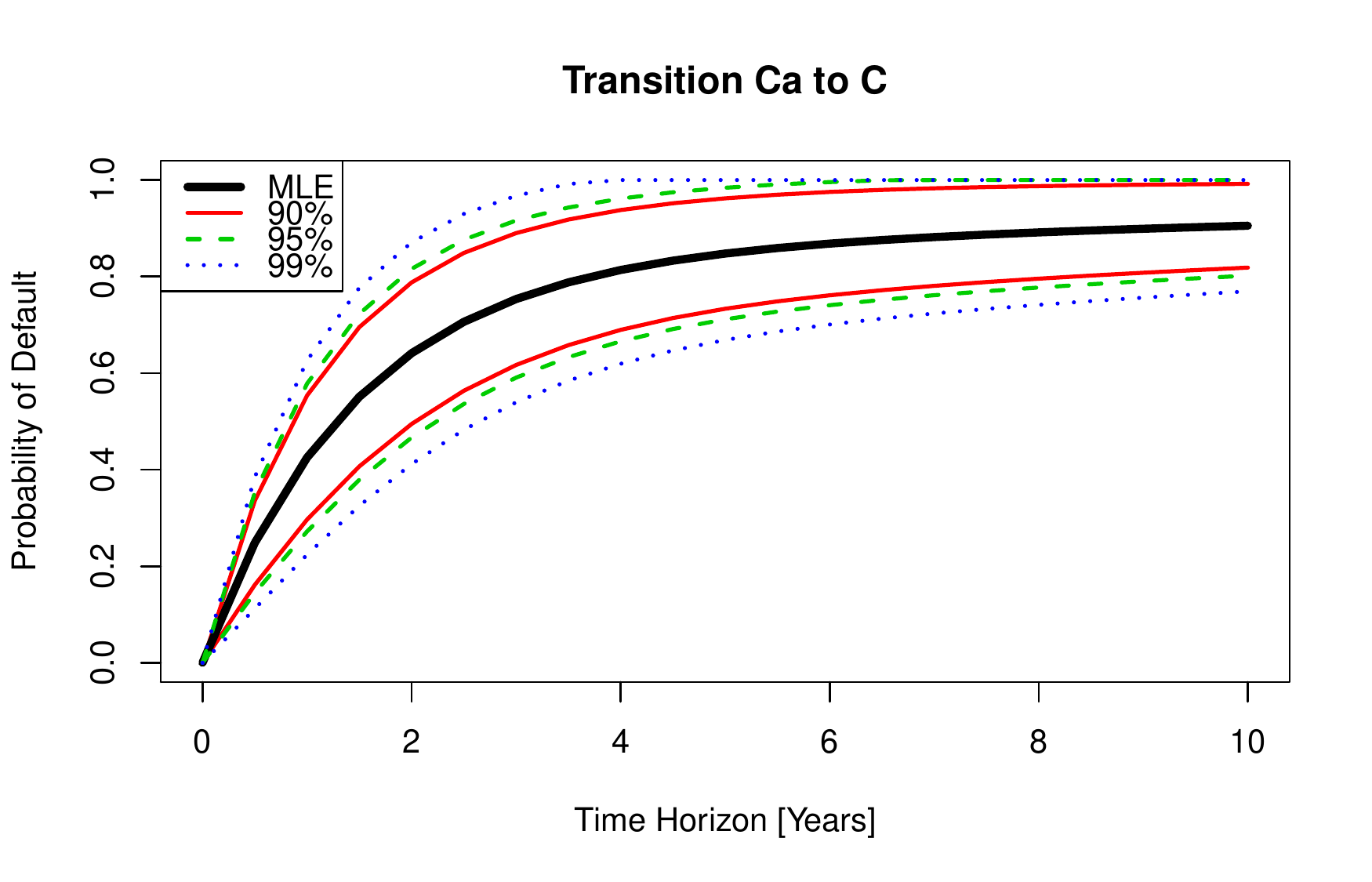}
\caption{Confidence Intervals as maps of time for Discrete-Time Transitions into the Default Category $C$ over 10 years-  Moody's Corporate Rating Discrete-Time Transitions 2016
}
\label{g2-2}
\end{figure}

\subsection{Confidence Intervals w.r.t.~information}
We benchmark our analysis against \cite[Section 4]{ReisSmith2017}. We consider a true generator matrix (which is the MLE Markov generator described in Section \ref{Sec:Examples}) and from that simulate multiple years worth of data which is viewed as empirical data. We then introduce the EM algorithm to increasing amounts of data and assess how the estimate and errors change. By using a known generator, we additionally assess the accuracy of the estimate and error. From a computational point of view, matrix exponentials embed highly nonlinear dependencies in the elements of $\mathbf{Q}$ and $\mathbf{P}$. Therefore, to understand the error we consider how both the error of $\mathbf{Q}$ and $\mathbf{P}$ changes as the amount of information changes.

We consider the scenario of 250 obligors per rating and simulate 50 years worth of transitions (i.e.~the number of companies that made each transition). We then apply the EM algorithm using $1$ year worth of data then 2 years etc up to 50 years. In the case of a company defaulting we replace it with the rating they were pre-default. This implies that the amount of ``information'' obtained from each year is similar. We plot the results in Figure \ref{Pic:EMTPMError}.
\begin{figure}[ht]
	\centering
	\includegraphics[width=15cm]{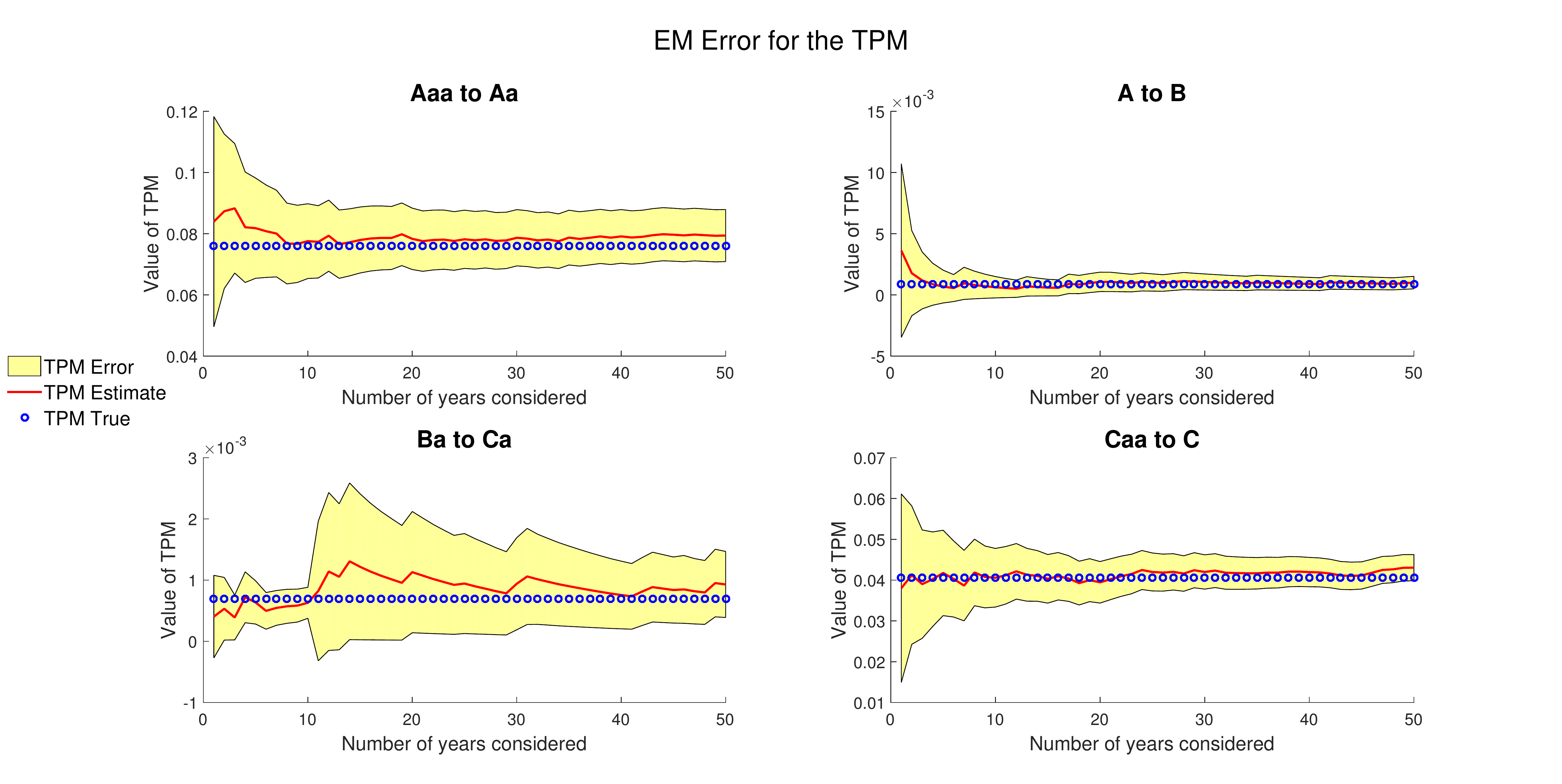}
	\vspace{-0.5cm}
	\caption{Estimated value in some TPM entries and 95\% confidence interval as the amount of data increases.}
	\label{Pic:EMTPMError}
\end{figure}

One observes that in most cases the errors in the TPM behave as expected. The surprising result is the $Ba$ to $Ca$ entry whose error increases. As alluded above, one can only understand the error in the TPM by understanding the underpinning error of the generator estimation. Although, in theory, the $Ba$ to $Ca$ transition depends on all entries in the generator we know that certain entries have a greater impact. We, therefore, look at the error in some important generator entries, Figure \ref{Pic:EMGeneratorError}.
\begin{figure}[ht]
	\centering
	\includegraphics[width=15cm]{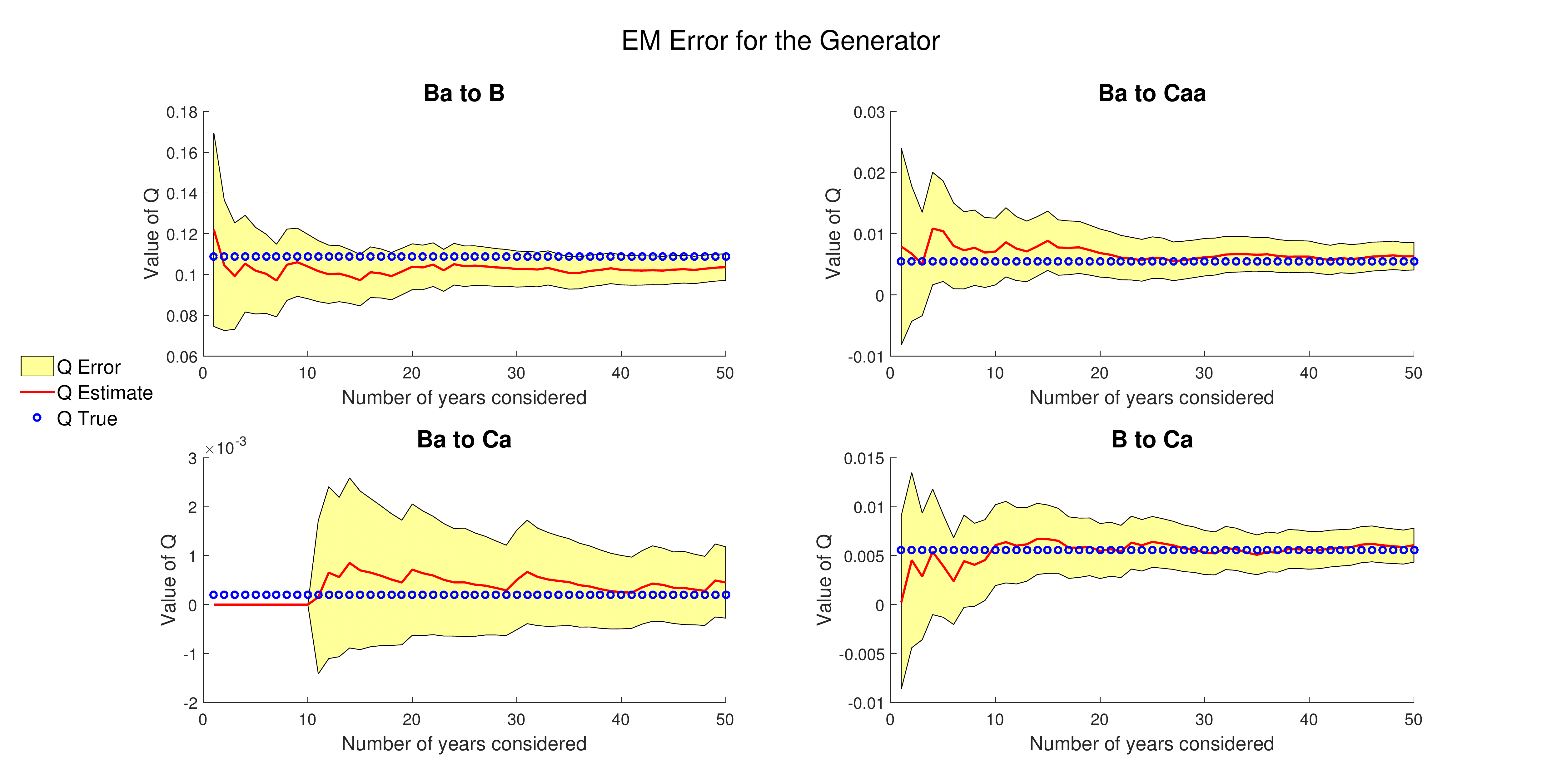}
	\vspace{-0.5cm}
	\caption{Estimated value in the generator and 95\% confidence interval as the amount of data increases.}
	\label{Pic:EMGeneratorError}
\end{figure}

From Figure \ref{Pic:EMGeneratorError} it is clear that the main contributor to the error is (unsurprisingly) the $Ba$ to $Ca$ entry. Initially, we need to wait for a transition from $Ba$ to $Ca$ to happen which increases the likelihood and hence the uncertainty surrounding the estimate. Moreover, it then takes several more years of data before the estimate becomes more stable. This uncertainty in the generator then propagates to uncertainty in the TPM entries, and one observes the extremely strong correlation between the TPM entry and the corresponding generator entry. Due to this, the error in the $Ba$ to $Ca$ transition probability is much larger than the other estimates, even after 50 years of observation. This behaviour in the CTMC modelling is not ideal (and the IFRS 9 regulation exacerbates the effect), but it shows some of the challenges in obtaining good estimates and errors for small probabilities (rare events), namely that the model is still sensitive to individual observations. One can use this to assess the sensitivity in the model, for example, adding one observation of a company defaulting and then recomputing the probabilities and their associated errors will provide an idea of the sensitivity.

%
%
%

%
%
%
%
%

\section{Extending Markov Processes to Capture Rating Momentum}
\label{Sec:Rating Momentum Model}

In this section, we work with the continuously observed data case and hence can broaden our scope of models (we are no longer restricted to Markov models). In the previous section, we highlighted many good features of the EM algorithm, in particular, that one could derive closed-form expressions for the errors. However, the EM algorithm does not generalize well as one quickly runs into difficulties when using models that have more complex likelihoods. This is the case when we generalize to point processes. Before detailing the model we are proposing let us start by showing that the data (see Section \ref{Sec:EM Error}) contains non-Markov features.

\subsection{Testing for non-Markovian phenomena}
In \cite{LandoSkodeberg2002}'s analysis of Standard and Poor's rating data set, the authors tested the presence of rating momentum. For consistency and completeness, we show that \emph{rating momentum} behaviour is also present in Moody's data set.

The test follows a standard semi-parametric hazard model approach developed in \cite{AndersenEtAl1991} (see also \cite{AndersenEtAl2012}). The basic idea is to test whether the intensity (from leaving the state) is influenced by previous transitions, that is, we model the intensity for any given firm, $n$ in state $i$ as,
\begin{align*}
\lambda_{i n}(t)= q_{i}(t) \exp(c Z_{n}(t)),
\end{align*}
where $q$ is an unspecified ``baseline'' intensity\footnote{Observe that we are not assuming that the baseline is time homogeneous in the test.}, $Z$ contains information relating to the firm and $c$ is the coefficient we estimate. One important point here is that we are often dealing with censored observations (many firms stop being rated after a while), hence using hazard models is useful since we have access to the theory of partial likelihoods which can handle censored observations, see \cite{CoxOakes1984}. One can then for example set the covariate $Z$ as,
\begin{align*}
Z_{n}(t)=
\begin{cases}
1, \quad &\text{if firm $n$ was downgraded to its current state,}\\
0, & \text{otherwise.}
\end{cases}
\end{align*}
Hence in this setting the Markov assumption is equivalent to the null hypothesis $c=0$. The general statistical framework including fitting $c$ by maximising the partial likelihood is covered in \cite{AndersenEtAl1991} and \cite[Appendix A]{LandoSkodeberg2002}, but we do not discuss these further here.

The result from this analysis can be seen in Table \ref{t23} -- we can see a statistically significant downward momentum effect (i.e. the null hypothesis is rejected on standard significance levels $\alpha$ of $10\%$, $5\%$ or $1\%$) but no significant upward momentum behaviour in the Moody's data. These findings are consistent with those of \cite{LandoSkodeberg2002}.

\begin{table}[!hbt]
	\centering
	\begin{tabular}{r|c|c}
		&coefficient& $p$-value\\
		\hline
		downward momentum &\phantom{-}0.33010 & $<$0.0001\\
		upward momentum   &-0.01487 &0.68153\\
		\hline
	\end{tabular}
	\caption{Likelihood ratio test for downward and upward momentum.}
	\label{t23}
\end{table}

\subsection{Our new Model to capture Rating Momentum}

As one can see from Table \ref{t23} there is very strong evidence that downward momentum exists in the data. Let us now describe a tractable methodology, using \emph{marked point processes} that can capture this effect. Readers unfamiliar with point processes can consult Appendix \ref{Sec:Overview of Point Processes} for further details. 

The likelihood of a single realisation of a marked point process is given in \cite[p.251]{DaleyVereJones2003}, namely,
\begin{equation}
L= \prod_{i=1}^{N_{g}(T)} \lambda_{g}(t_{i})f(k_{i}|t_{i})e^{-\int_{0}^{T}\lambda_{g}(u) \dd u} \, , \label{Eq:Point process likelihood}
\end{equation}
where we use the following notation, $N_{g}$ is the set of times at which events occur, $\lambda_{g}$ is the intensity, $k$ is the mark and $f$ is the so-called \emph{mark's distribution}. The subscript $g$ is a common notation used to imply that this is the intensity of the ground process, i.e.~we are only considering the events of interest. 
Setting $\lambda=q_{i}$ and $f=q_{ij}/q_{i}$ we recover the likelihood of a CTMC and hence one can see that these processes are generalizations of Markov processes.

To incorporate rating momentum into such models we draw inspiration from Hawkes processes and change the intensity of the model for appropriate rating changes. The basic idea is to start with a CTMC (with generator matrix $\mathbf{Q}$), which acts as a baseline intensity, then add a non-Markov component which is a self-excitation intensity decaying exponentially\footnote{This is a common and well-understood form to use in Hawkes processes, see \cite{BacryMastromatteoMuzy2015}.}. That is, any downgrade observed increases the intensity of then future downgrades for a certain while. We also introduce two types of momentum, one if the company downgrades from investment-grade ($Baa$ and better) and another if the company downgrades from a speculative-grade (this modelling choice is further discussed in Section \ref{Sec:Bayesian Information} and \ref{Sec:Examples}).  
Using the same notation as before, given the state space $\{1, \dots, h\}$ such that state $h$ (default) is absorbing, we model the intensity of the stochastic process $X$ at time $t$ as follows,
\begin{equation*}
\lambda_{g}(t)= \sum_{j=1}^{h-1}q_{j} \1_{\{X(t)=j\}} + \sum_{m=1}^{2}\sum_{\tau \in \tau_{m}(t)} \beta_{m} \alpha_{m} e^{-\beta_{m}(t-\tau)} \, ,
\end{equation*}
where $m$ denotes investment or speculative downgrade, $\tau_{m}(t)$ is the set of downgrade times (of type $m$) prior to time $t$ and $\alpha_{m}$ and $\beta_{m}$ correspond to the intensity and memory of the ``momentum'' in each case. One can note that the intensity of the stochastic process drops (returning to the baseline intensity) as more time elapses since the previous downgrade and this rate is controlled by ${\bm \beta}$. In particular, this allows one to include empirically observed effects such as the momentum's influence reducing over time (see \cite{couderc2008credit}). 

In this set up we add only four parameters to the $\approx (h-1)^{2}$ parameters of the CTMC case; the effectiveness of this parsimony is substantiated below (see Section \ref{Sec:Bayesian Information}). To the best of our knowledge, no other model we are aware of captures the momentum effect so simply. Further parameters and extensions can be introduced, nonetheless, we focus only on this model. Its analysis is found in Section \ref{Sec:Model Calibration} and \ref{Sec:Examples}.

We work under the following modelling assumptions which, we believe to be sufficiently reasonable and keep the model parsimonious (most of these can be easily lifted and the model extended).
\begin{enumerate}
	\item We only consider downward momentum. Since upward momentum is not as statistically significant (Table \ref{t23}) we do not consider it.
	
	\item There are two types of momentum, \emph{investment} and \emph{speculative}. 
	
	Companies being downgraded from investment grades (numerically these are the ratings from $1$ to $(h-1)/2$) feel the investment momentum and remaining downgrades are affected by speculative momentum.
	
	\item Finally (not easy to remove) no points occurred prior to time $0$, the so-called edge effects. This essentially means that companies do not have momentum when they are initially rated.
\end{enumerate}

\begin{remark}
	[Prudent Estimation]
	Since we only consider momentum as a purely negative effect, if we assume a company has no momentum when it initially does then we will obtain more conservative numbers for the downgrades. Therefore in calibration, if one does not use a full history of a company's rating change the model will be more prudent.
\end{remark}

With these assumptions let us define the mark's distribution. We take the following marked distribution (for $X(t_{i}) \in \{1, \dots, h-1\}$, $t_i$ is the time of the $i$th jump),
\begin{align*}
f(X(t^{-}_{i})|t_{i})
& =
\dfrac{\sum_{j,k=1}^{h}q_{j k}\1_{\{X(t^{-}_{i})=j, ~ X(t_{i})=k\}}}{\lambda_{g}(t_{i})} 
\Big(
\1_{\{X(t_{i})<X(t^{-}_{i})\}}
\\
&
\qquad \qquad \qquad +
	\1_{\{X(t_{i})>X(t^{-}_{i})\}}
	\frac{1}{N_{j}} \sum_{m=1}^{2} \sum_{\tau \in \tau_{m}(t_{i})} \beta_{m} \alpha_{m} e^{-\beta_{m} (t_{i}-\tau)} 
	\Big) \, ,	
\end{align*}
where we denote by $t_{i}^{-}$ the time immediately prior to the $i$th jump and $N_{j}$ is the number of states one can downgrade to i.e.~$N_{j}= \sum_{k>j} \1_{\{q_{jk}>0 \}}$. Substituting the intensity and mark distribution into \eqref{Eq:Point process likelihood}, yields the following expression for the likelihood,
\begin{align}
L=
\prod_{i=1}^{N_{g}(T)} 
&\Biggl\{
\bigg(
\sum_{j,k=1}^{h}q_{j k}\1_{\{X(t^{-}_{i})=j, ~ X(t_{i})=k\}}
+
\frac{1}{N_{j}} \sum_{m=1}^{2} \sum_{\tau \in \tau_{m}(t_{i})} \beta_{m} \alpha_{m} e^{-\beta_{m} (t_{i}-\tau)}
\bigg)
\1_{\{X(t_{i})>X(t^{-}_{i})\}}
\notag
\\
&
+
 \sum_{j,k=1}^{h}q_{j k}\1_{\{X(t^{-}_{i})=j, ~ X(t_{i})=k\}}
 \1_{\{X(t_{i})<X(t^{-}_{i})\}}
\Biggr\}
\notag
\\
&
\qquad \times
 \exp\left( -\int_{0}^{T}\sum_{j=1}^{h-1}q_{j} \1_{\{X(u)=j\}}
                   + \sum_{m=1}^{2}\sum_{\tau \in \tau_{m}(u)} \beta_{m} \alpha_{m} e^{-\beta_{m}(u-\tau)} \dd u
				   \right) \, .
\label{Eq:Full Likelihood Momentum}
\end{align}
Note that the likelihood is for the information regarding one company. We can construct the likelihood of multiple companies by taking the product, but it is worthwhile noting that this assumes independence among companies. This is unlikely to be true due to business cycles etc, however, these correlated systemic effects can be introduced into risk modelling using the methods from \cite{McNeilWendin2007}. Hence, we concentrate purely on the idiosyncratic effect of rating momentum.

The integral involving the momentum (last integral in \eqref{Eq:Full Likelihood Momentum}) can be simplified, to
\begin{align*}
\int_{0}^{T}\sum_{\tau \in \tau_{m}(u)} \beta_{m} \alpha_{m} e^{-\beta_{m}(u-\tau)} \dd u
=
\sum_{\tau \in \tau_{m}(T)} \alpha_{m} \left(1- e^{-\beta_{m}(T-\tau)}\right) \, .
\end{align*}
Unlike the CTMC case this likelihood is complex and there appears to be no real simplification, the main reason for this is the time and history dependence amongst jumps for which simplifications of the form $q_{ij}^{\mathbf{K}_{ij}}$ are no longer possible. We proceed forward by relying  on Markov Chain Monte Carlo (MCMC) techniques to estimate the parameters.

\subsection{An MCMC calibration algorithm for the model}

 In the CTMC setting as considered in \cite{BladtSorensen2005}, \cite{BladtSorensen2009} and \cite{ReisSmith2017} the data augmentation step for the CTMC was costly making the algorithm extremely slow compared to other algorithms. In our setting, we have access to a complete data set and this expensive step is avoided. Moreover, the likelihood we deal with is complex and thus MCMC (see \cite{GilksRichardsonEtAl1996}) is one of the few methods that can deliver reasonable estimations. 

The basic set up of MCMC is to estimate the parameter(s) $\theta$ through its posterior distribution given some data $D$, typically denoted $\pi(\theta| D)$. In general, one cannot access this posterior distribution and direct Monte Carlo simulation is not possible as one does not know the normalizing constant. MCMC gets around this by observing through Bayes' formula that,
\begin{align*}
\pi(\theta|D) \propto L(D; \theta) \pi(\theta) \, ,
\end{align*}
where $L$ is the likelihood and $\pi(\theta)$ is the prior distribution of $\theta$. It is then possible to sample from this distribution using the Metropolis-Hastings algorithm with some proposal distribution.

Let $\mathbf{X}$ denote the set of all company transitions. We are interested in obtaining the joint distribution $\pi(\mathbf{Q},{\bm \alpha}, {\bm \beta}|\mathbf{X})$ where $\mathbf{Q}$ is the matrix with the baseline intensities and jump probabilities (has the same form as a generator matrix of a CTMC) and ${\bm \alpha}:=(\alpha_{1}, \alpha_{2})$, ${\bm \beta}:= (\beta_{1}, \beta_{2})$ are the momentum parameters. Since we assume the prior distribution of $\mathbf{Q}$, ${\bm \alpha}$ and ${\bm \beta}$ to be independent, Bayes' theorem implies that,
\begin{equation*}
\pi(\mathbf{Q},{\bm \alpha},{\bm \beta}|\mathbf{X}) \propto \pi( \mathbf{X}| \mathbf{Q},{\bm \alpha},{\bm \beta})\pi(\mathbf{Q})\pi({\bm \alpha})\pi({\bm \beta}) = L\pi(\mathbf{Q})\pi({\bm \alpha})\pi({\bm \beta}) \, ,
\end{equation*}
where $L$ is the likelihood defined in \eqref{Eq:Full Likelihood Momentum}. The full conditional distribution of each parameter is obtained by conditioning on knowledge of all other parameters.

For the priors, firstly for $\mathbf{Q}$, we assume that the initial transitions carry no momentum hence we can set the prior as the CTMC maximum likelihood estimate (MLE) based on the initial transitions. We therefore set the prior as exponential with the mean being the MLE. For ${\bm \alpha}$ and ${\bm \beta}$, we use a Gamma random variable with a reasonable variance as the prior. This is to reflect that we have far less knowledge for these parameters but do not expect them to be either zero or too large.

The next issue we tackle is how to simulate from the full conditional distribution. Dealing with the parameters of the model first, their full conditional distributions are clearly not standard distributions so we use the single-component Metropolis-Hastings algorithm. As always with Metropolis-Hastings we need to define a good proposal function. In order to avoid a high number of rejections, we take our proposal as a Gamma random variable with mean as the current step and a small variance. In effect, this creates a random walk type sampling scheme that is always nonnegative. Therefore, if we denote the set of parameters by $\gamma$ and the proposal distribution by $\psi$ (which can depend on the current parameters), the $n$th step acceptance probability of a proposed point $\gamma_{s}$ given the current $\gamma_{s}'$ is given by,
\begin{equation*}
\frac{\pi(X|\gamma_{s}, \gamma_{n, -s})\pi(\gamma_{s}) \psi(\gamma_{s}'| \gamma_{s})}{\pi(X|\gamma_{s}', \gamma_{n, -s})\pi(\gamma_{s}') \psi(\gamma_{s}| \gamma_{s}')} \, ,
\end{equation*}
where $\gamma_{n, -s}$ denotes the set of parameters at the $n$th update not including the $s$ parameter.

\subsubsection{Model Calibration}
\label{Sec:Model Calibration}
Now that we have the necessary tools, we can calibrate our model using Moody's data set. Running $11000$ MCMC iterations (taking $1000$ burn in) we obtain the following results\footnote{The MCMC algorithm, written in MATLAB, took $\approx 8.5$ hours to run on a Intel Xeon E7-4660 v4 2.2GHz processor.}. For the Markov style ``base'' component,
\small
\begin{align*}
\mathbf{Q}=
\left(
\begin{array}{ccccccccc}
\text{Aaa} & \text{Aa}  & \text{A} & \text{Baa}  & \text{Ba}  & \text{B}  & \text{Caa}  &  \text{Ca}  & \text{C} 
\\
 -0.0869 &    0.0836 &   0.0031 &  0       & 0.0002    &  0     & 0     & 0   &   0
 \\
 0.0117  & -0.1088   & 0.0942   & 0.0025   & 0.0003  &  0.0001   &      0 &        0  &    0
 \\
 0.0006  &  0.0240   & -0.0938  &  0.0666  &  0.0017 &   0.0007  &  0.0002 &       0   &    0
 \\
 0.0002  &  0.0016   & 0.0387   & -0.0947  &  0.0496 &   0.0040  &  0.0006 &   0.0000  &       0
 \\
 0.0001  &  0.0006   & 0.0033   & 0.0636   & -0.1774 &   0.1060  &  0.0037 &   0.0001  &       0
 \\
 0.0000  &  0.0003   & 0.0012   & 0.0035   & 0.0503  & -0.1610   & 0.1012  &  0.0040   & 0.0004
 \\
 0  &  0.0002   & 0.0001   & 0.0013   & 0.0048  &  0.1028 &  -0.1976  &  0.0622   & 0.0261
 \\
 0   &      0   & 0.0018   & 0.0029   & 0.0050  &  0.0447 &   0.1346  & -0.2838   & 0.0948
 \\
 0    &     0    &     0   &      0    &     0  &      0  &       0   &      0    &     0
 \end{array}
 \right),
 \end{align*}
\normalsize
and for the momentum parameters,
\begin{align*}
{\bm \alpha}= (0.031, 0.1291 )
\qquad
\text{and}
\qquad
{\bm \beta}=(3.5234, 1.7095).
\end{align*}

One interesting observation arising from calibration is the difference of momentum parameters across the investment and the speculative downgrades. There is apparently more momentum in the speculative downgrades than in the investment downgrades, namely, the momentum intensity is larger and lasts longer in speculative grades\footnote{Note that both $(0.1 \approx) \,\alpha_{1}\beta_{1}< \alpha_{2}\beta_{2} \,(\approx 0.2)$ and $\beta_{1}>\beta_{2}$.}. 

This may seem counter-intuitive, however, setting a credit rating ultimately involves combining information from various sources and making a judgement on the exposure of that company (sovereign) to different risks. As discussed in \cite[Chapter 5 and 6]{couderc2008credit}, there appears to be a noticeable difference on which information influences downgrades/defaults for investment-grade and speculative-grade obligors. This points towards an intrinsic difference between these classes of ratings and thus it is not too surprising that our momentum model also shows a difference. From a practical point of view, the model suggests that a downgrade in a speculative-grade company is more damming for future performance, the information that influences speculative-grade rating changes implies deeper issues within the company and hence higher chances of further downgrades/default.

\subsection{Bayesian Information Criterion}
\label{Sec:Bayesian Information}

Let us give some justification for the use of this model. We have argued that a point process style model is a strong choice and to keep the model as robust and simple as possible we added four extra ``momentum parameters'' (with relation to the CTMC model). We believe four to be the optimal choice due to the fact that only adding two parameters does not yield as good a fit to the observed data and adding parameters to every rating does not seem appropriate, since we do not have enough transitions across all ratings to obtain a reliable fit. We therefore did not consider more momentum groups than investment and non-investment grade.

As we have access to a full data set, one can directly calculate the MLE for the $\mathbf{Q}$ generator matrix of the Markov model setting. Therefore we can test our momentum model against the purely Markov model.

The Markov model is a particular case of our momentum model, set $\alpha_{i}=0$ and $\beta_{i}$ a constant for $i \in {1,2}$. Hence, a priori the non-Markov model stands to fit the data better (in the sense of achieving a likelihood at least as large). The question we look to answer is, are we actually capturing the data better or just overfitting? To do this we calculate the Bayesian Information Criterion (BIC), it is a common test used in statistics for model selection and is known to penalize model complexity more than other statistical tests, such as the Akaike information criterion (see \cite[Chapter 3]{ClaeskensHjort2008}). We believe this feature makes the BIC a good test to justify our more complex model. The BIC for a model $M$ can be written as (some authors use the negative of this)
\begin{align*}
\text{BIC}(M)= 2 \log\big( L(M|D) \big) - \log(n) \text{dim}(M) \, ,
\end{align*}
where $n$ refers to the number of data points and $\text{dim}(M)$ is the number of parameters in the model. From a given set of models, the model with the largest BIC is taken as the better one. Naturally, the indicator of how much ``better'' one model is over another is the difference in the BIC, where a BIC difference strictly greater than $10$ is taken as very strong evidence of the model superiority.

\begin{table}[!ht]
	\centering
	\small
	\begin{tabular}{ c | c }
		~  &  BIC
		\\ \hline
		Difference & $138.5 \gg 10$
		\\ \hline
	\end{tabular}
	\caption{The BIC difference between the non-Markov and Markov model on the Moody's dataset.}
	\label{Table:Model Selection}
\end{table}

The result in Table \ref{Table:Model Selection} gives us confidence that our non-Markov model captures reality better without overfitting and with sufficient parsimony with relation to the Markov (CTMC) one.

\subsection{Examples and testing}
\label{Sec:Examples}

\emph{Probabilities of default as maps of time: Markov Vs.~non-Markov.} One important aspect of the non-Markov theory is how it impacts the estimates for the TPM and the transition probabilities.

\begin{remark}
	[Obtaining transition probabilities and model simulation]
	In the standard Markov set up, the TPM is calculated using \eqref{Eq:P and Q relation}. In the non-Markov set up we do not have such a simple relation, hence we are forced to use Monte Carlo techniques. In this case, we prescribe multiple companies in each rating at the start (we used a total of $10^{7}$) and simulate individual transitions according to the point process model. By recording the rating of each company at various points in time (see below)  we can then build transition matrices over several time horizons in the same way one builds an empirical TPM. 
	
	The simulation of our momentum model is similar to that of a standard CTMC, i.e. based on the current state one simulates a ``jump time'' then simulates the new state to jump into. The main difference here is the added complexity of the time and history dependence that exists in our momentum model. To simulate the jump time of a fixed company we use the standard accept/reject method introduced in \cite[Algorithm 2]{Ogata1981} for varying intensities. For each accepted jump time, we then calculate the transition probabilities based on this time and simulate the jump to the next state. We then repeat this process for each company up until the time horizon required.
\end{remark}

It is of particular interest to understand how the evolution in time of the probabilities of default change when using the CTMC Markovian and our non-Markovian model. Using the calibrated model, Figure \ref{Fig:ProbOfDefaults} details the probabilities of defaults for the various ratings as maps in time.

\begin{figure}[!ht]
	\centering
	\includegraphics[width=15cm]{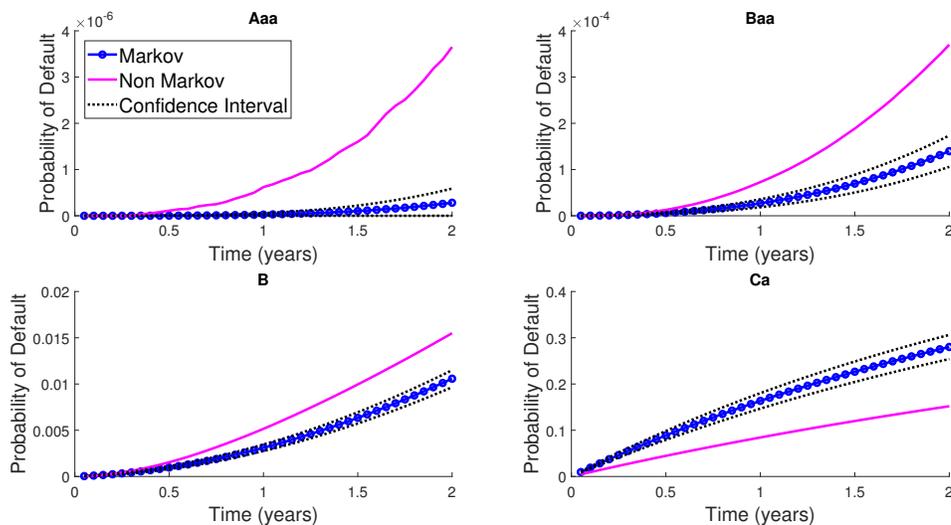}
	\vspace{-0.5cm}
	\caption{The probability of default given by each model for various ratings as a function of time.
}
	\label{Fig:ProbOfDefaults}
\end{figure}

The first observation one can make from Figure \ref{Fig:ProbOfDefaults} is, the non-Markov model produces higher probabilities of default, except for the $Ca$ rating (the non-Markov default probability is also lower for rating $Caa$). The reason for this is precisely the non-Markovianity in the data. In a Markov framework, all companies in the same rating are treated the same, consequently, it is unlikely that an investment-grade company will continue to downgrade quickly while the non-Markov model allows for this.

On the other hand, companies may enter rating $Ca$ before defaulting, hence in the momentum model, some companies in this rating are carrying an extra term making default more likely. This implies we can account for a larger number of defaults while keeping the $\mathbf{Q}$ matrix $Ca$ to $C$ entry smaller. This is not the case in the Markov model and thus to produce enough defaults from $Ca$ one makes the $\mathbf{Q}$ matrix entry larger. Consequently, the Markov model overestimates the default probability for obligors initially rated $Ca$.

\emph{Probabilities of default: Empirical Vs.~Markov Vs.~non-Markov.}
To test how reliable these results are, we can compare one-year probabilities of default as estimated from each calibrated model compared to that we observe from the data. To do so, we fix some time horizon $T$ (one-year here) and consider all companies that have either defaulted or not withdrawn by this period. We then build an empirical TPM over this horizon based on the company's rating at time zero and $T$. Concentrating solely on probabilities of default we obtain the results in Table \ref{Table:Results from Algorithms}.

\begin{table}[!ht]
	\centering
	\footnotesize
	\begin{tabular}{ c | c |  c | c | c | c | c | c  |c }
\multirow{2}{*}{\backslashbox{Model}{Ratings}}  & \multicolumn{4}{c|}{Investment-Grade} & \multicolumn{4}{c}{Speculative-Grade}
	\\
	\cline{2-9}
		& Aaa & Aa & A & Baa & Ba & B & Caa & Ca
		\\ \hline
		Empirical  & 0 & 0 & 0 & 0.0004000 & 0.0005 & 0.0012 & 0.0064 & 0.0563
		\\
		non-Markov & $1\times 10^{-6}$ & $4\times 10^{-6}$ & 0.0000125 & 0.0000734 & 0.0011 & 0.0052 & 0.0298 & 0.0845
		 \\
		 Markov & $3\times 10^{-8}$ & $2.5\times 10^{-7}$ & $4.86\times 10^{-7}$ & 0.0000271 & 0.0002 & 0.0031 & 0.0407 & 0.1635 \\
		\hline
		\hline
		{\# Companies per}
		  & \multirow{2}{*}{413} & \multirow{2}{*}{1313} & \multirow{2}{*}{2232} & \multirow{2}{*}{2318} & \multirow{2}{*}{2021} & \multirow{2}{*}{4504} & \multirow{2}{*}{1333} & \multirow{2}{*}{59} \\		
			rating at $t=0$ &  &  &  &  &  &  &  &  \\		
		\hline
	\end{tabular}
	\caption{Comparing one-year probability of defaults of each model against the empirical observations. For reference we add explicitly the number of companies per rating at starting time $t=0$ as in the Moody’s data set.}
	\label{Table:Results from Algorithms}
\end{table}

The results in Table \ref{Table:Results from Algorithms} are interesting because they highlight stark differences in the models. Starting with the investment-grade, unfortunately, we do not have enough data to fully assess default probabilities at this level. The only grade for which a default within a year is observed is $Baa$ and is higher than what both models predict. One reason the momentum model may not capture this probability as well is the way we have set up the momentum parameters, i.e. an investment and speculative set, and $Baa$ is at the turning point. On the other hand, this number is estimated from a smaller number of defaults so is subject to a larger error. Comparing the Markov and non-Markov, it is unsurprising that our model makes investment-grade defaults more likely.

For the speculative-grades, one observes that $Ca$ and $Caa$ firms have lower one-year default probabilities in the non-Markov model and these estimates are closer to the empirical observations. This is exactly due to the reason mentioned previously, companies downgrading into $Ca$ and $Caa$ ``poison'' the data in the Markov setting.  
Implying that in a Markov world a company initially rated $Caa$ or $Ca$ is viewed to be riskier than it actually is. 

The difference between the models may have a large impact on a bank's capital requirements for regulation. Although the non-Markov model makes most ratings riskier than the Markov model, we feel it provides a more accurate reflection of default risk.

\begin{remark}
[Limitations from censored data]
	Unfortunately, in our study, we are limited to small-time horizons due to censored data. Namely, since the default is absorbing, as soon as a company defaults, we keep that information up to the terminal time. However, many companies are only rated over a few years before withdrawing and therefore if we look at empirical TPMs over longer horizons they are built with less (non-default) data. Since we do not want to use the Markov assumption, there does not appear to be a way to incorporate this lost data. Therefore we can only obtain ``accurate'' numbers on short time scales.
\end{remark}

%
%
%
%

\section{Summary}

In the first part of this paper we have shown how one can evaluate errors in the transition matrices of continuous-time Markov chains at the level of discretely observed data using new closed-form expressions. These results reduced the computation of confidence intervals to less than one half of the time needed by current approaches. Moreover, and of practical importance, by employing the Delta method, our results provide an intuitively interpretable understanding of uncertainty in the model output, the probabilities of default.

In the second part, we have shown the significance of being able to capture non-Markov effects in rating transitions. Comparing against empirical probabilities of default and the classical Markov chain model, one finds a tendency for the Markov chain model to overestimate on some speculative-grades and underestimate on investment-grades. We address this issue by providing a parsimonious model that better captures default probabilities (where empirically observed). Moreover, the non-Markov model points towards significantly higher probabilities of default for investment-grades, where such values are not empirically observed, thus making it more prudent. We believe that the model we present provides a more accurate view of reality and hence should be considered in credit risk modelling. These observations further highlight the importance of understanding so-called \emph{model risk} and its potential impact in quantitative risk analysis in general.

\section*{Acknowledgements and Funding}
The authors would like to thank Dr.~R.~P.~Jena at Nomura Bank plc London, Prof.~M.~Fischer at Bayerische Landesbank Munich, also H.~Thompson and M.~Chen from the Zurich Insurance Group and Dr.~M.~de Carvalho at the University of Edinburgh for the helpful comments. The authors would like to thank two anonymous referees for their comments and suggestions to clarify the manuscript.
\medskip
\small

G. dos Reis acknowledges support from the \emph{Funda{\c c}$\tilde{\text{a}}$o para a Ci$\hat{e}$ncia e a Tecnologia} (Portuguese Foundation for Science and Technology) through the project [UID/MAT/00297/2019] (Centro de Matem\'atica e Aplica\c c$\tilde{\text{o}}$es CMA/FCT/UNL).

M.~Pfeuffer acknowledges funding by DZ Bank Foundation and Universit\"atsbund Erlangen-Nuremberg (grant [S020/10264/17]).

G.~Smith was supported by The Maxwell Institute Graduate School in Analysis and its Applications, a Centre for Doctoral Training funded by the UK Engineering and Physical Sciences Research Council (grant [EP/L016508/01]), the Scottish Funding Council, the University of Edinburgh and Heriot-Watt University. The views expressed in the paper are solely those of the author and do not represent the views of his employer, Moody’s Analytics, its parent company (Moody’s Corporation) or its affiliates. 

\small


\bibliographystyle{rQUF}


\begin{thebibliography}{0}
\providecommand{\natexlab}[1]{#1}
\providecommand{\noopsort}[1]{}
\providecommand{\printfirst}[2]{#1}
\providecommand{\singleletter}[1]{#1}
\providecommand{\switchargs}[2]{#2#1}

\end{thebibliography}


\begin{thebibliography}{38}
\providecommand{\natexlab}[1]{#1}
\providecommand{\noopsort}[1]{}
\providecommand{\printfirst}[2]{#1}
\providecommand{\singleletter}[1]{#1}
\providecommand{\switchargs}[2]{#2#1}

\bibitem[\protect\citeauthoryear{{Altman} and {Kao}}{1992}]{AltmanKao1992}
{Altman}, E. and {Kao}, D.L., The implications of corporate bond ratings drift.
  {\itshape Financial Analysts Journal}, 1992, \textbf{48}, 64--75.

\bibitem[\protect\citeauthoryear{Andersen
  {\itshape{et~al.}}}{2012}]{AndersenEtAl2012}
Andersen, P.K., Borgan, O., Gill, R.D. and Keiding, N., {\itshape Statistical
  models based on counting processes}, 2012, Springer Science \& Business
  Media.

\bibitem[\protect\citeauthoryear{Andersen
  {\itshape{et~al.}}}{1991}]{AndersenEtAl1991}
Andersen, P.K., Hansen, L.S. and Keiding, N., Non-and semi-parametric
  estimation of transition probabilities from censored observation of a
  non-homogeneous {M}arkov process. {\itshape Scandinavian Journal of
  Statistics}, 1991, pp. 153--167.

\bibitem[\protect\citeauthoryear{Bacry
  {\itshape{et~al.}}}{2015}]{BacryMastromatteoMuzy2015}
Bacry, E., Mastromatteo, I. and Muzy, J.F., Hawkes processes in finance.
  {\itshape Market Microstructure and Liquidity}, 2015, \textbf{1}, 1550005.

\bibitem[\protect\citeauthoryear{Bielecki
  {\itshape{et~al.}}}{2011}]{BieleckiCrepeyHerbertsson2011}
Bielecki, T.R., Cr{\'e}pey, S. and Herbertsson, A., Markov chain models of
  portfolio credit risk. In {\itshape The Oxford Handbook of Credit
  Derivatives}, chap.~10, pp. 327--382, 2011, Oxford University Press.

\bibitem[\protect\citeauthoryear{Bladt and
  S{\o}rensen}{2005}]{BladtSorensen2005}
Bladt, M. and S{\o}rensen, M., Statistical inference for discretely observed
  {M}arkov jump processes. {\itshape Journal of the Royal Statistical Society:
  Series B (Statistical Methodology)}, 2005, \textbf{67}, 395--410.

\bibitem[\protect\citeauthoryear{Bladt and
  S{\o}rensen}{2009}]{BladtSorensen2009}
Bladt, M. and S{\o}rensen, M., Efficient estimation of transition rates between
  credit ratings from observations at discrete time points. {\itshape
  Quantitative Finance}, 2009, \textbf{9}, 147--160.

\bibitem[\protect\citeauthoryear{Capp\'{e}
  {\itshape{et~al.}}}{2005}]{CappeEtAl2005}
Capp\'{e}, O., Moulines, E. and Ryd\'{e}n, T., {\itshape Inference in hidden
  {M}arkov models}, Springer Series in Statistics, 2005, Springer, New York,
  With Randal Douc's contributions to Chapter 9 and Christian P. Robert's to
  Chapters 6, 7 and 13, With Chapter 14 by Gersende Fort, Philippe Soulier and
  Moulines, and Chapter 15 by St\'{e}phane Boucheron and Elisabeth Gassiat.

\bibitem[\protect\citeauthoryear{Carey and Hrycay}{2001}]{CareyHrycay2001}
Carey, M. and Hrycay, M., Parameterizing credit risk models with rating data.
  {\itshape Journal of Banking \& Finance}, 2001, \textbf{25}, 197 -- 270.

\bibitem[\protect\citeauthoryear{Christensen
  {\itshape{et~al.}}}{2004}]{ChristensenEtAl2004}
Christensen, J.H., Hansen, E. and Lando, D., Confidence sets for
  continuous-time rating transition probabilities. {\itshape Journal of Banking
  \& Finance}, 2004, \textbf{28}, 2575--2602.

\bibitem[\protect\citeauthoryear{Claeskens and
  Hjort}{2008}]{ClaeskensHjort2008}
Claeskens, G. and Hjort, N.L., Model selection and model averaging. {\itshape
  Cambridge Books}, 2008.

\bibitem[\protect\citeauthoryear{Couderc}{2008}]{couderc2008credit}
Couderc, F., Credit Risk and Ratings: Understanding Dynamics and Relationships
  with Macroeconomics. PhD thesis, Ecole Polytechnique F{\'e}d{\'e}rale de
  Lausanne, 2008.

\bibitem[\protect\citeauthoryear{Cox and Oakes}{1984}]{CoxOakes1984}
Cox, D.R. and Oakes, D., {\itshape Analysis of survival data}, 1984, Routledge.

\bibitem[\protect\citeauthoryear{Daley and
  Vere-Jones}{2003}]{DaleyVereJones2003}
Daley, D.J. and Vere-Jones, D., {\itshape An introduction to the theory of
  point processes. {V}ol. {I}}, Second , Probability and its Applications (New
  York), 2003, Springer-Verlag, New York, Elementary theory and methods.

\bibitem[\protect\citeauthoryear{Daley and
  Vere-Jones}{2008}]{DaleyVereJones2007}
Daley, D.J. and Vere-Jones, D., {\itshape An introduction to the theory of
  point processes. {V}ol. {II}}, Second , Probability and its Applications (New
  York), 2008, Springer, New York, General theory and structure.

\bibitem[\protect\citeauthoryear{D'Amico {\itshape{et~al.}}}{2016}]{DAmico2016}
D'Amico, G., Janssen, J. and Manca, R., Downward migration credit risk problem:
  a non-homogeneous backward semi-{M}arkov reliability approach. {\itshape
  Journal of the Operational Research Society}, 2016, \textbf{67}, 393--401.

\bibitem[\protect\citeauthoryear{Dassios and Zhao}{2013}]{DassiosZhao2013}
Dassios, A. and Zhao, H., Exact simulation of {H}awkes process with
  exponentially decaying intensity. {\itshape Electronic Communications in
  Probability}, 2013, \textbf{18}, 1--13.

\bibitem[\protect\citeauthoryear{{dos Reis} and Smith}{2018}]{ReisSmith2017}
{dos Reis}, G. and Smith, G., Robust and consistent estimation of generators in
  credit risk. {\itshape Quantitative Finance}, 2018, \textbf{18}, 983--1001.

\bibitem[\protect\citeauthoryear{Frydman and
  Schuermann}{2008}]{FrydmanSchuermann2008}
Frydman, H. and Schuermann, T., Credit rating dynamics and {M}arkov mixture
  models. {\itshape Journal of Banking \& Finance}, 2008, \textbf{32},
  1062--1075.

\bibitem[\protect\citeauthoryear{Gilks
  {\itshape{et~al.}}}{1996}]{GilksRichardsonEtAl1996}
Gilks, W.R., Richardson, S. and Spiegelhalter, D.J., Introducing {M}arkov Chain
  {M}onte {C}arlo. {\itshape {M}arkov chain {M}onte {C}arlo in practice}, 1996,
  \textbf{1}, 19.

\bibitem[\protect\citeauthoryear{Inamura}{2006}]{Inamura2006}
Inamura, Y., Estimating continuous time transition matrices from discretely
  observed data. Technical report, Citeseer, 2006.

\bibitem[\protect\citeauthoryear{Kalbfleisch and
  Lawless}{1985}]{KalbfleischLawless1985}
Kalbfleisch, J. and Lawless, J.F., The analysis of panel data under a {M}arkov
  assumption. {\itshape Journal of the American Statistical Association}, 1985,
  \textbf{80}, 863--871.

\bibitem[\protect\citeauthoryear{Knight}{2000}]{Knight2000}
Knight, K., {\itshape Mathematical statistics}, Chapman \& Hall/CRC Texts in
  Statistical Science Series, 2000, Chapman \& Hall/CRC, Boca Raton, FL.

\bibitem[\protect\citeauthoryear{Koopman
  {\itshape{et~al.}}}{2008}]{KoopmanEtAl2008}
Koopman, S.J., Lucas, A. and Monteiro, A., The multi-state latent factor
  intensity model for credit rating transitions. {\itshape Journal of
  Econometrics}, 2008, \textbf{142}, 399--424.

\bibitem[\protect\citeauthoryear{Korolkiewicz}{2012}]{Korolkiewicz2012}
Korolkiewicz, M.g.W., A dependent hidden {M}arkov model of credit quality.
  {\itshape Int. J. Stoch. Anal.}, 2012, pp. Art. ID 719237, 13.

\bibitem[\protect\citeauthoryear{Kreinin and
  Sidelnikova}{2001}]{KreininSidelnikova2001}
Kreinin, A. and Sidelnikova, M., Regularization algorithms for transition
  matrices. {\itshape Algo Research Quarterly}, 2001, \textbf{4}, 23--40.

\bibitem[\protect\citeauthoryear{Lando and
  Skodeberg}{2002}]{LandoSkodeberg2002}
Lando, D. and Skodeberg, T.M., Analyzing rating transitions and rating drift
  with continuous observations. {\itshape Journal of Banking and Finance},
  2002, \textbf{26}, 423--444.

\bibitem[\protect\citeauthoryear{Lehmann and
  Casella}{1998}]{LehmannCasella1998}
Lehmann, E.L. and Casella, G., {\itshape Theory of point estimation}, Second ,
  Springer Texts in Statistics, 1998, Springer-Verlag, New York.

\bibitem[\protect\citeauthoryear{L{\o}ffler}{2005}]{Loffler2005}
L{\o}ffler, G., Avoiding the rating bounce: why rating agencies are slow to
  react to new information. {\itshape Journal of Economic Behavior \&
  Organization}, 2005, \textbf{56}, 365--381.

\bibitem[\protect\citeauthoryear{McNeil
  {\itshape{et~al.}}}{2005}]{McNeilFreyEmbrechts2005}
McNeil, A.J., Frey, R. and Embrechts, P., {\itshape Quantitative Risk
  Management: Concepts, Techniques and Tools}, 2005  (Princeton: Oxford).

\bibitem[\protect\citeauthoryear{McNeil and Wendin}{2007}]{McNeilWendin2007}
McNeil, A.J. and Wendin, J.P., Bayesian inference for generalized linear mixed
  models of portfolio credit risk. {\itshape Journal of Empirical Finance},
  2007, \textbf{14}, 131--149.

\bibitem[\protect\citeauthoryear{Nickell
  {\itshape{et~al.}}}{2000}]{NickellPerraudinVarotto2000}
Nickell, P., Perraudin, W. and Varotto, S., Stability of rating transitions.
  {\itshape Journal of Banking \& Finance}, 2000, \textbf{24}, 203--227.

\bibitem[\protect\citeauthoryear{Oakes}{1999}]{Oakes1999}
Oakes, D., Direct calculation of the information matrix via the {EM}. {\itshape
  Journal of the Royal Statistical Society: Series B (Statistical
  Methodology)}, 1999, \textbf{61}, 479--482.

\bibitem[\protect\citeauthoryear{Ogata}{1981}]{Ogata1981}
Ogata, Y., On {L}ewis' simulation method for point processes. {\itshape
  Information Theory, IEEE Transactions on}, 1981, \textbf{27}, 23--31.

\bibitem[\protect\citeauthoryear{Ogata}{1988}]{Ogata1988}
Ogata, Y., Statistical models for earthquake occurrences and residual analysis
  for point processes. {\itshape Journal of the American Statistical
  association}, 1988, \textbf{83}, 9--27.

\bibitem[\protect\citeauthoryear{Pfeuffer}{2017}]{Pfeuffer2017-R-Package}
Pfeuffer, M., {ctmcd: An R Package for Estimating the Parameters of a
  Continuous-Time Markov Chain from Discrete-Time Data}. {\itshape {The R
  Journal}}, 2017, \textbf{9}, 127--141.

\bibitem[\protect\citeauthoryear{{Van Loan}}{1978}]{VanLoan1978}
{Van Loan}, C., Computing integrals involving the matrix exponential. {\itshape
  Automatic Control, {IEEE} Transactions on}, 1978, \textbf{23}, 395--404.

\bibitem[\protect\citeauthoryear{Wilcox}{1967}]{Wilcox1967}
Wilcox, R., Exponential operators and parameter differentiation in quantum
  physics. {\itshape Journal of Mathematical Physics}, 1967, \textbf{8},
  962--982.

\end{thebibliography}



%
%
%
%

\appendix

\section{Fundamentals of Discretely Observed Markov Processes}
\label{Sec:Fundamentals of Markov}

The EM algorithm for this problem (See Section \ref{Sec:EM Error}) is discussed in detail in \cite{BladtSorensen2005} and \cite{ReisSmith2017} and we encourage the reader to consult these texts for further information. For completeness, we present a brief review of the EM algorithm for the setting of continuous-time Markov chains. 

For convergence of the EM algorithm, one works under the following assumption. 
\begin{assumption}[Element constraint]
	Similar to \cite{BladtSorensen2005}, we will use a manual space constraint to obtain the convergence. Take $1>\epsilon > 0$, such that for $i\neq j$, $q_{ij}< 1/\epsilon$. Moreover, we assume adjacent mixing, namely, for $i \in \{2,\dots, h-1\}$, $q_{i, i\pm 1}>\epsilon$ and $q_{1,2}> \epsilon$.  
	
We denote the space of generator matrices which satisfy this condition as $\Lambda_{\epsilon}$.
\end{assumption}
This assumption is a trivial constraint when one works in credit risk as it requires that: (a) firms can be upgraded or downgraded by one rating which is clearly the case; and (b) that changes in ratings do not happen too fast which is also the practical case. 

Let $(X(t))_{t\geq 0}$ be a stochastic process over the finite state space $\{1, \dots, h\}$. Associated to $X(t)$ is, for $i,j$ in the state space, $\mathbf{K}_{ij}(t)$ the number of jumps from $i$ to $j$ in the interval $[0,t]$ and by $\mathbf{S}_{i}(t)$ the holding time of state $i$ in the interval $[0,t]$. The EM algorithm is then given by,
\begin{enumerate}[(i)]
	\item Take an initial intensity matrix $\mathbf{Q}$ and a small positive value $\epsilon$, so $\mathbf{Q} \in \Lambda_{\epsilon}$.
	
	\item While the convergence criterion is not met and $\mathbf{Q} \in \Lambda_{\epsilon}$,
	\begin{enumerate}[(1)]
		\item E-step: calculate $\bE_{\mathbf{Q}}[\mathbf{K}_{ij}(T)| \mathbf{P}]$ and $\bE_{\mathbf{Q}}[\mathbf{S}_{i}(T)|\mathbf{P}]$.
		\item M-step: set $q'_{ij}=\bE_{\mathbf{Q}}[\mathbf{K}_{ij}(T)|\mathbf{P}]/\bE_{\mathbf{Q}}[\mathbf{S}_{i}(T)|\mathbf{P}]$, for all $i \neq j$ and set $q_{ii}$ appropriately.
		\item Set $\mathbf{Q}=\mathbf{Q}'$ (where $\mathbf{Q}'$ is the matrix of $q'$s) and return to E-step.
	\end{enumerate}
	
	\item End while and return $\mathbf{Q}$.
\end{enumerate}

By \cite[Theorem 2.10]{ReisSmith2017}, provided the algorithm does not hit the boundary of $\Lambda_{\epsilon}$, we obtain convergence (in distribution and parametric) to a stationary point.
Typically the E-step in the EM algorithm needs to be calculated numerically, however \cite{ReisSmith2017} following \cite{VanLoan1978} and \cite{Inamura2006} obtained the following result.
\begin{proposition}
	Let $e_{i}$ be the column vector of length $h$ which is one at entry $i$ and zero elsewhere, further let us define the $2h$-by-$2h$ matrices $\mathbf{C}^{(\alpha \beta)}_{\gamma}$ and $\mathbf{C}^{(\alpha)}_{\phi}$ as,
	\begin{equation*}
	\mathbf{C}^{(\alpha \beta)}_{\gamma} :=
	\left[ {\begin{array}{cc}
		\mathbf{Q} & q_{\alpha \beta} \mathbf{e}_{\alpha} \mathbf{e}_{\beta}^{\intercal} \\
		0 & \mathbf{Q}
		\end{array} } \right]
	\qquad
	\text{ and }
	\qquad
	\mathbf{C}^{(\alpha)}_{\phi} :=
	\left[ {\begin{array}{cc}
		\mathbf{Q} & \mathbf{e}_{\alpha} \mathbf{e}_{\alpha}^{\intercal} \\
		0 & \mathbf{Q}
		\end{array} } \right]\,
	\quad \alpha,\beta\in\{1,\cdots,h\}.
	\end{equation*}
	Consider a CTMC $X$ observed at $n$ time points $0 \le t_1 < t_2 < \dots < t_n$; denote by $y_s$ the state of the chain at time $t_s$, i.e.~$y_s := X(t_s)$.
	Then, the expected jumps and holding times across observations are,
	\begin{align*}
	\bE_{\mathbf{Q}}[\mathbf{K}_{ij}(t)|\mathbf{y}]
	=
	&
	\sum_{s=1}^{n-1} \frac{\left(\exp(\mathbf{C}^{(i j)}_{\gamma} (t_{s+1} -t_{s}))\right)_{y_{s},h+y_{s+1}}}{\left(\exp(\mathbf{Q} (t_{s+1} -t_{s}))\right)_{y_{s},y_{s+1}}} \, ,
	\\
	\bE_{\mathbf{Q}}[\mathbf{S}_{i}(t)|\mathbf{y}]=
	&
	\sum_{s=1}^{n-1} \frac{\left(\exp(\mathbf{C}^{(i)}_{\phi} (t_{s+1} -t_{s}))\right)_{y_{s},h+y_{s+1}}}{\left(\exp(\mathbf{Q} (t_{s+1} -t_{s}))\right)_{y_{s},y_{s+1}}}.
	\end{align*}
\end{proposition}
When one only has access to an observed sequence of TPMs $\mathbf{P}$ with equal observation length we obtain,
\begin{align*}
\bE_{\mathbf{Q}}[\mathbf{K}_{ij}(T)|\mathbf{P}]
=
& \sum_{u=1}^{M} \sum_{s=1}^{h}\sum_{r =1}^{h} \mathbf{P}^{u}_{sr}(t)\frac{\left(\exp(\mathbf{C}^{(i j)}_{\gamma} t)\right)_{s,h+r}}{\left(\exp(\mathbf{Q} t)\right)_{s,r}} \, , \notag
\\
\bE_{\mathbf{Q}}[\mathbf{S}_{i}(T)|\mathbf{P}]
=
&
\sum_{u=1}^{M} \sum_{s=1}^{h}\sum_{r=1}^{h} \mathbf{P}^{u}_{sr}(t) \frac{\left(\exp(\mathbf{C}^{(i)}_{\phi}t)\right)_{s,h+r}}{\left(\exp( \mathbf{Q} t)\right)_{s,r}} \, ,
\end{align*}
where $M=T/t$ (the number of observations) and $\mathbf{P}^{u}$ is the TPM of the $u$-th observation.

Roughly speaking, the above formula is taking each row in the TPM to contain equal amounts of information (observations). When one knows the number of transitions between the states $\mathbf{N}$, then $\mathbf{P}_{sr}^{u}(t)$ is replaced by $\mathbf{N}_{sr}(u)$, where $\mathbf{N}_{sr}(u)$ is the number of observed transitions in observation $u$.

The $M$-step is just the ratio of these two quantities and thus the results yield closed-form expressions for the EM algorithm's steps making the algorithms much faster (see results in \cite{ReisSmith2017}).

\section{Proof of Theorem \ref{Thm:CTMC Delta Method}}
\label{Sec:Proof of Delta Method}

The proof relies on the multivariate delta method, see \cite[Theorem 8.16]{LehmannCasella1998}.
\begin{proposition}[Delta Method]
	\label{Prop:Multivariate Delta}
	Let $(X_{1 \nu}, \dots, X_{s \nu})$, $\nu=1, \dots, n$, be $n$ independent $s$-tuples of random variables with $\bE [ X_{i \nu} ] = \xi_{i}$ and $\text{cov}(X_{i \nu},X_{j \nu})= \sigma_{ij}$. Let $\bar{X}_{i}$ denote the empirical mean, $\bar{X}_{i} := \sum_{\nu} X_{i \nu}/n$, and suppose that $h$ is a real-valued function of $s$ arguments with continuous first partial derivatives. Then,
	\begin{align*}
	\sqrt{n}\Big( h(\bar{X}_{1}, \dots, \bar{X}_{s})-h(\xi_{1}, \dots, \xi_{s})\Big) \xrightarrow{\text{Dist}}
	\cN(0, v^2)
	,
	\quad
	v^{2}= \sum_{i} \sum_{j} \sigma_{ij} \frac{\partial h}{\partial \xi_{i}} \frac{\partial h}{ \partial \xi_{j}} \, , \quad \text{provided } v^{2}>0.
	\end{align*}
\end{proposition}

We now have the necessary material to prove our result.
\begin{proof}[Proof of Theorem \ref{Thm:CTMC Delta Method}]
	The assumption of asymptotic normality implies the expectation and covariance assumption of Proposition \ref{Prop:Multivariate Delta}. Moreover, it follows from standard results in likelihood based inference that ${\bm \sigma} \approx -  \mathbf{H}(\hat{\mathbf{Q}})^{-1}$ (see \cite[Chapter 5.4]{Knight2000}).

	For the partial derivatives of the probability matrix, it follows immediately by arguments in Section \ref{Sec:Direct Differentiation}.
	Also note that this representation implies that the first partial derivatives of $p_{ij}$ exist and are continuous.
	
To complete the proof, we need only to show that the RHS of \eqref{Eq:Variance estimate} is strictly positive.  Firstly, at a maximum $\mathbf{H}$ is negative definite (hence $\mathbf{H}^{-1}$ is also negative definite), therefore it is enough to have that $\partial p_{ij}/ \partial \mathbf{V}_{\hat{\mathbf{Q}}} \neq 0$ around the MLE. Observing the latter is one of the theorem's assumptions concludes the proof.
\end{proof}

\section{Overview of Point Processes}
\label{Sec:Overview of Point Processes}

Let us discuss how we look to embed history dependence into the model. We are interested in Hawkes processes (a specific type of \emph{self-exciting} point process) which have intensities of the form
\begin{align*}
\lambda_{t} =
\mu
+
\int_{0}^{t} \phi(t-s) \dd N_{s} \, .
\end{align*}
Hawkes processes are used to model many different phenomena, from earthquake occurrence to high frequency trading, see \cite{Ogata1988} and \cite{BacryMastromatteoMuzy2015}.  
Setting $\phi =0$ yields a constant intensity and this is equivalent to the Markov setting. However, $\phi$ allows us to vary the intensity with past events which is key for \emph{momentum} since past downgrades influence future transitions. As described in the introduction, a Hawkes process is just a counting process (generalising a Poisson process), hence it would imply that a rating transition had occurred and not which rating we have moved into. The latter being key, we consider processes which take values on some state space: such processes are known as \emph{marked point process (MPPs)}, see \cite[Section 6.4]{DaleyVereJones2003}. MPPs are point processes on a product space $\cT \times \cK$, that is, we return a set of values, $\{t_{k}, \kappa_{k}\}$ for $k=1,2, \dots$, where one thinks of $t_{k}$ as the event time of the point process (with intensity $\lambda$) and $\kappa_{k}$ of the ``mark'' associated to the event. These notions are what we shall use, but in general $t_{k}$ can be multidimensional e.g.~to include spatial dependence. In our case we have $\kappa_{k} \in \{1, \dots, h\}$ namely, it denotes the ratings which ensure our marked point process to be well defined.

The likelihood of a single realisation of a MPP, is given in \cite[p.251]{DaleyVereJones2003},
\begin{equation*}
L= \prod_{i=1}^{N_{g}(T)} \lambda_{g}^{*}(t_{i})f^{*}(\kappa_{i}|t_{i})e^{-\int_{0}^{T}\lambda_{g}^{*}(u) \dd u} \, ,
\end{equation*}
where we have the following notation, $N_{g}$ is the set of events occur, $\lambda_{g}$ is the intensity and $f$ is the so-called \emph{mark's distribution}. The $^{*}$ symbolises that the intensity and mark distribution depend on previous events. Namely, the intensity at time $t_{i}$, $\lambda_{g}^{*}(t_{i})$ depends on the previous events, $\{(t_{1}, \kappa_{1}), \dots, (t_{i-1}, \kappa_{i-1})\}$. Also note the distinction that $\lambda_{g}^{*}(t_{i})$ does not depend on the mark $\kappa_{i}$, but the mark $\kappa_{i}$ is allowed to depend on time $t_{i}$. The subscript $g$ is a common notation used to imply this is the ground process, which in our case is simply the timing of the upgrades/downgrades. By allowing the intensity and hence the number of jumps and the mark distribution to depend on previous events we can easily change the probability of upgrade/downgrade and thus embed rating momentum into the process.
Further details on likelihoods of MPP can be found in \cite[Section 7.3]{DaleyVereJones2003}.

One reason that we believe MPPs are a good choice for this particular problem is that one can view them as a natural generalisation of CTMCs. This is apparent from the likelihood since, letting $\lambda=q_{i}$ and $f=q_{ij}/q_{i}$ we recover the likelihood of a CTMC.

\end{document}